\documentclass[11pt,a4paper,oneside]{article}

\usepackage[top=2cm, bottom=2cm, left=2cm, right=2cm]{geometry}
\usepackage[english]{babel}
\usepackage[utf8]{inputenc}
\usepackage[T1]{fontenc}

\usepackage{amsthm,amsmath,amssymb,amsfonts,mathrsfs,dsfont,mathtools}
\usepackage{bm,bbm}
\usepackage{cancel}

\usepackage{graphicx}
\usepackage{caption}
\usepackage{subcaption}
\usepackage{booktabs}
\usepackage{array}
\usepackage{multirow}
\usepackage{hhline}
\usepackage{threeparttable}
\usepackage[figuresright]{rotating}
\usepackage{arydshln}

\usepackage{algorithm}
\usepackage{algpseudocode}

\usepackage[table]{xcolor}
\usepackage[colorinlistoftodos]{todonotes}
\usepackage{comment}
\usepackage{verbatim}
\usepackage{dirtytalk}

\usepackage{enumitem}
\usepackage{appendix}
\usepackage{anyfontsize}
\usepackage{lipsum}

\usepackage{natbib}
\usepackage[colorlinks=true, allcolors=blue]{hyperref}

\newtheorem{proposition}{Proposition}

\DeclarePairedDelimiter\floor{\lfloor}{\rfloor}

\newcommand{\Hline}{%
  \noalign{\hrule height 0.4pt}%
  \noalign{\vskip 1pt}%
  \noalign{\hrule height 0.4pt}%
}

\newcommand{\simind}{\stackrel{\mbox{\tiny ind}}{\sim}}
\newcommand{\simiid}{\stackrel{\mbox{\tiny iid}}{\sim}}

\newcommand{\btheta}{\bm{\theta}}

\newcommand{\by}{\bm{y}}
\newcommand{\bY}{\bm{Y}}

\newcommand{\bw}{\bm{w}}
\newcommand{\bz}{\bm{z}}
\newcommand{\bdelta}{\bm{\delta}}
\newcommand{\bpsi}{\bm{\psi}}

\newcommand{\opt}{\textsc{opt}}

\newcommand{\dd}{\mathrm{d}}
\renewcommand{\Pr}{\mathrm{Pr}}

\begin{document}

\title{Uncertainty quantification and multi-stage variable selection\\ for personalized treatment regimes}
\author{Jiefeng Bi, Matteo Borrotti, and Bernardo Nipoti 
\\
Department of Economics, Management and Statistics, University of Milano-Bicocca, Italy}

\allowdisplaybreaks



\maketitle

\begin{abstract}
A dynamic treatment regime is a sequence of medical decisions that adapts to the evolving clinical status of a patient over time. To facilitate personalized care, it is crucial to assess the probability of each available treatment option being optimal for a specific patient, while also identifying the key prognostic factors that determine the optimal sequence of treatments. This task has become increasingly challenging due to the growing number of individual prognostic factors typically available. In response to these challenges, we propose a Bayesian model for optimizing dynamic treatment regimes that addresses the uncertainty in identifying optimal decision sequences and incorporates dimensionality reduction to manage high-dimensional individual covariates. The first task is achieved through a suitable augmentation of the model to handle counterfactual variables. For the second, we introduce a novel class of spike-and-slab priors for the multi-stage selection of significant factors, to favor the sharing of information across stages. The effectiveness of the proposed approach is demonstrated through an extensive simulation study and illustrated using clinical trial data on severe acute arterial hypertension.
\end{abstract}


\section{Introduction} 
Dynamic treatment regimes (DTRs) have emerged as a popular framework for implementing personalized care, offering adaptive medical decisions that adjust over time in response to a patient's evolving clinical state \citep{murphy2003optimal, murphy2005experimental}. By exploiting patient-specific information, DTRs aim to recommend optimal treatments at each stage of care, proving especially effective for chronic diseases and complex conditions. The task of constructing optimal DTRs based on patient data can be viewed as a multi-stage decision-making problem, akin to a learning agent navigating an environment by selecting actions that maximize rewards. This framework is rooted in reinforcement learning, a branch of machine learning that has provided numerous solutions to multi-stage decision problems, from dynamic programming \citep{Bel57} to Q-learning \citep{Wat89}. In the statistical literature, Q-learning has been adapted to address the specific challenges of DTRs \citep{zhao2009reinforcement}. Building on this concept, many contributions on DTRs \citep[see][and references therein]{chakraborty2013statistical,tsiatis2019dynamic} have employed regression-based models, using 
backward induction to account for the dynamic nature of the problem. The optimal treatment regime is then defined as the sequence of treatments that maximizes the overall expected response for an individual, also known as the individual value. The literature on Bayesian methods for DTRs has also advanced significantly over the past two decades. A predictive framework was introduced by \citet{arjas2010, arjas2012} and later extended by \citet{saarela2015bayesian, saarela2016doubly}. Foundational work from a decision-theoretic perspective appears in \citet{dawid2010overview, dawid2015statistical}. A comprehensive overview is provided by \citet{oganisian2021}, with recent advancements in \citet{rodriguez2022, rodriguez2023}. A coherent Bayesian framework for regression-based models was introduced by \citet{murray2018bayesian}. Within this framework, we propose a modeling strategy called Bayesian augmented learning, which is designed to effectively quantify the uncertainty involved in identifying optimal treatment decision sequences, a fundamental task that has been largely overlooked in the existing literature. A second objective of our proposal is to perform variable selection within the context of DTRs, a crucial step given the increasing availability of prognostic factors for individual patients. In regression-based approaches to DTRs, most variable selection methods have been developed within the frequentist paradigm \citep{qian2011performance, 
song2015penalized, fan2016sequential, wallace2019model, zhu2019proper, bian2021variable}.
 Although less explored in the DTR literature, Bayesian methods offer a compelling alternative, as they naturally accommodate the quantification of uncertainty regarding the significance of prognostic factors. We address this gap by pairing our modeling approach with a Bayesian variable selection strategy tailored to the dynamic structure of DTRs, achieved by adapting the class of spike-and-slab priors \citep[see][]{Tad21} to this setting. 
\subsection{Our contribution}
The specific contributions of this work can be summarized as follows.
\begin{itemize}
\item[i)] We propose a modeling strategy called Bayesian augmented learning, which addresses the presence of counterfactual variables in DTRs through a suitable augmentation. 
\item[ii)] We introduce the dependent spike-and-slab prior for variable selection, taking into account the association between the significance of the same covariate across multiple stages. 
\item[iii)]We conduct a comprehensive empirical study to evaluate the performance of the proposed modeling framework and compare it with state-of-the-art methods for estimating DTRs.
\item[iv)] We demonstrate our approach using data from the MIMIC-III clinical database \citep{johnson2016mimic}, which includes multiple treatment options and numerous prognostic factors.
\end{itemize} 
The paper proceeds as follows. Section~\ref{section2} outlines our modeling approach. Section~\ref{section4} assesses its performance through simulation studies, while Section~\ref{section5} illustrates its application to MIMIC-III data. Section~\ref{section6} concludes with a discussion. Proofs, computational details, and additional results from the data analyses are provided in the Appendix.

\section{Bayesian augmented learning with variable selection}\label{section2}
We consider a two-stage DTR framework where $a_j\in\mathcal{A}_j$ indicates the treatment, or action, assigned at the $j$th stage, with the index $j=1,2$ henceforth used to refer to specific stages. We let $Y=Y_1+Y_2$ denote the overall payoff, with $Y_j$ being the part of the payoff observed after the $j$th stage. Dealing with specific applications in which $Y_1$ is not observable is possible by letting $Y_1=0$ and $Y=Y_2$. Hereafter, the index $i = 1, \ldots, n$ denotes quantities related to the $i$th statistical unit, e.g. $a_{ij}$ for the treatment assigned to unit $i$ at stage $j$, and similarly for $Y_i$ and $Y_{ij}$. In order to stress the dependence of a payoff on the assigned treatments $a_{i1}$ and $a_{i2}$, we write $Y_i(a_{i1},a_{i2}) = Y_{i1}(a_{i1})+Y_{i2}(a_{i1},a_{i2})$. 
Throughout this work, capital letters denote random variables, while lowercase letters represent their realizations. For instance, the observed payoffs $y_i$ are realizations of the corresponding random variables $Y_i(a_{i1},a_{i2})$. 
We further introduce the notation $Y_i^{\opt}$ to denote the overall payoff $Y_{i}(a_{i1},a_{i2}^\opt)$, where $a_{i2}^\opt$ is a treatment that, given that the treatment $a_{i1}$ was assigned at stage 1 and according to an optimal decision rule yet to be specified, is optimal at stage 2. Following the principle of backward induction \citep[see, e.g.,][]{tsiatis2019dynamic}, we consider a sequence of two 
regression models, one for the second-stage payoffs and one for the optimal overall payoffs. 
Namely, 
\begin{equation}\label{eq:model1}
\begin{aligned}
   Y_{i2}
    \mid a_{i1}, a_{i2},\bm{z}_{i2},\btheta_2 &\simind f_2(y\mid a_{i1},a_{i2},\bm{z}_{i2};\btheta_2)\\
    Y_{i}^{\opt}
    \mid a_{i1},\bm{z}_{i1},\btheta_1 &\simind f_1(y\mid a_{i1},\bm{z}_{i1};\btheta_1),
    \end{aligned}
\end{equation}
where $\bm{z}_{ij}$ denotes the vector of covariates available for the $i$th individual at the $j$th stage, $f_1$ and $f_2$ are regression functions, $\btheta_1$ and $\btheta_2$ are vectors of regression coefficients. To facilitate the identification of causal effects of stage-specific actions, we henceforth 
adopt the following standard assumptions: the stable unit treatment value assumption \citep{rubin1978bayesian} and the assumption of no unmeasured confounders \citep{rubin1980bias}. See Appendix \ref{sec:appASS} 
for more details and \citet{tsiatis2019dynamic} for a thorough discussion. Model \eqref{eq:model1} is completed by specifying an optimal decision rule. Following a large body of literature on DTR optimization, we adopt a dynamic programming approach \citep{Bel57} and define the optimal treatment regime as the sequence of treatments that maximizes the expected payoff for an individual, assuming, without loss of generality, that higher payoffs are preferable. Accordingly, the optimal treatments at stage 2 and stage 1 are defined as
\begin{equation}\label{eq:decision_rule}
    a_{i2}^{\opt}=\underset{a\in\mathcal{A}_2}{\arg\sup}(\mathbb{E}[Y_{i2}\mid a_{i1},a,\bm{z}_{i2},\btheta_2]),\quad
    a_{i1}^{\opt}=\underset{a\in\mathcal{A}_1}{\arg\sup}(\mathbb{E}[Y_{i}\mid a,a_{i2}^\opt,\bm{z}_{i1},\btheta_1]).
\end{equation} 
For the sake of simplicity, we assume that the solution to the optimization problems in \eqref{eq:decision_rule} is unique: this will be the case, with probability 1, for the specification of $f_2$ and $f_1$ considered throughout the paper. A well-known hurdle characterizing DTR models, such as model \eqref{eq:model1}, is that, while, given the assigned treatments $a_{i1}$ and $a_{i2}$, the payoff $Y_{i}=Y_{i}(a_{i1},a_{i2})$ is an observable quantity, $Y_{i}^{\opt}=Y_{i}(a_{i1},a_{i2}^\opt)$ in general is not. $Y_{i}^\opt$ coincides with the observable $Y_{i}(a_{i1},a_{i2})$ when $a_{i2}^\opt=a_{i2}$, and is otherwise counterfactual, highlighting the need for missing data imputation. We propose a novel strategy to handle both observed and counterfactual payoffs. For simplicity, we illustrate it assuming binary actions ($\mathcal{A}_1 = \mathcal{A}_2 = \{0,1\}$), though it naturally extends to any finite set of treatments, as done in Sections~\ref{PART2} and~\ref{section5}. We assume that, at stage 2, the observed outcomes $Y_{i2}$ and the counterfactual payoffs $\bar{Y}_{i2} = Y_{i2}(a_{i1}, 1 - a_{i2})$ are conditionally independent given $\btheta_2$. Likewise, the overall payoffs $Y_{i}^\opt$ and $\bar{Y}^\opt_{i} = Y_{i}(1 - a_{i1}, a_{i2}^\opt)$ are conditionally independent given $\btheta_1$. We further assume that the regression functions $f_2$ and $f_1$ in \eqref{eq:model1} satisfy the following equivalences:
\begin{equation}\label{eq:alt_opt1}
    a_{i2}^\opt=a_{i2} \iff \frac{\Pr(Y_{i2}\geq\bar{Y}_{i2})}{\Pr(Y_{i2}\leq\bar{Y}_{i2})}> 1,\quad    
    a_{i1}^\opt=a_{i1} \iff \frac{\Pr(Y_{i}^\opt\geq\bar{Y}_{i}^\opt)}{\Pr(Y_{i}^\opt\leq\bar{Y}_{i}^\opt)}> 1.
\end{equation}
The validity of such an assumption is key and will be discussed in Section~\ref{sec:reg}. The double implications in \eqref{eq:alt_opt1} suggest an alternative 
strategy to implement the optimal rule \eqref{eq:decision_rule}: rather than maximizing the expected values in \eqref{eq:decision_rule}, a stage-specific optimal treatment can be identified as the one that is more likely to lead to a larger payoff. To implement this idea, we introduce 
the augmenting random vectors $V_{i2}^{(0:1)}=(V_{i2}^{(0)},V_{i2}^{(1)})$ and $V_{i}^{(0:1)}=(V_{i}^{(0)},V_{i}^{(1)})$, referred to as stage 2 and overall pseudo-outcomes, which, conditionally on $\btheta_2$ and $\btheta_1$, are independent of and 
distributed as $(Y_{i2}(a_{i1},0),Y_{i2}(a_{i1},1))$ and $(Y_{i}(0,a_{i2}^{\opt}),Y_{i}(1,a_{i2}^{\opt}))$, respectively. The pseudo-outcomes can be interpreted as vectors of potential payoffs of individuals with identical covariate information as those for which observations were actually collected. 
We thus define the random variables $A_{i2}^\opt=\mathbbm{1}{\{V_{i2}^{(1)}>V^{(0)}_{i2}\}}$ and $A_{i1}^\opt=\mathbbm{1}{\{V^{(1)}_{i}>V_{i}^{(0)}\}}$,
and observe that the mode of $A_{ij}^\opt$ is $a_{ij}^\opt$, while its expectation coincides with $\Pr(V_{i2}^{(1)}>V^{(0)}_{i2})$ if $j=2$, and  $\Pr(V_{i}^{(1)}>V^{(0)}_{i})$ if $j=1$. The joint distribution of observations $Y_{i2}$ and pseudo-outcomes $V_{i2}^{(0:1)}$ defines an augmented version of the second-stage regression in \eqref{eq:model1}; likewise, $Y_i^\opt$ and $V_i^{(0:1)}$ define an augmented version of the second regression. Moreover, if we assume that 
    $Y_i^\opt$ is equal in distribution to $Y_{i1}+\max(V_{i2}^{(0:1)})$, 
    we can write the following model:
    \begin{equation}\label{eq:aug_model}
        \begin{aligned}
 (Y_{i2},V_{i2}^{(0:1)})
    \mid a_{i1}, a_{i2},\bm{z}_{i2},\btheta_2 &\simind f_2(y\mid a_{i1},a_{i2},\bm{z}_{i2};\btheta_2)\prod_{a=0}^1 f_2(v_a\mid a_{i1},a,\bm{z}_{i2};\btheta_2)\\
   (Y_{i1},V_{i}^{(0:1)})\mid V_{i2}^{(0:1)},a_{i1},\bm{z}_{i1},\btheta_1 &\simind f_1(y+\max(V_{i2}^{(0:1)})\mid a_{i1},\bm{z}_{i1};\btheta_1) \prod_{a=0}^1
   f_1(v_a \mid a,\bm{z}_{i1};\btheta_1),
    \end{aligned}
    \end{equation}
    completed by assigning a prior $p(\btheta_1,\btheta_2)$ to the regression coefficients. The introduction of pseudo-outcomes leads to an efficient computational scheme for posterior simulation, offering three key advantages: 
i) counterfactual observations are treated as latent random variables and their imputation is seemingly handled; ii) a single computational scheme allows to sample from the joint posterior distribution of all stage-specific parameters, avoiding two-step procedures as in 
\citet{murray2018bayesian}; iii) the uncertainty in identifying individual optimal treatments can be quantified through the posterior distribution of $A_{i2}^\opt$ and $A_{i1}^\opt$.\\
We refer to the use of the model \eqref{eq:aug_model} to address DTR problems as \emph{Bayesian augmented learning} (BAL). The BAL approach handles $T \geq 2$ treatment options by defining $T$-dimensional stage 2 and overall pseudo-outcomes, denoted $V_{i2}^{(0:(T-1))} = (V_{i2}^{(0)}, \ldots, V_{i2}^{(T-1)})$ and $V_i^{(0:(T-1))} = (V_{i}^{(0)}, \ldots, V_{i}^{(T-1)})$, respectively. Key to the BAL model is the specification of the regression functions $f_2$ and $f_1$, along with the prior $p(\btheta_1, \btheta_2)$, which are detailed in Sections~\ref{sec:reg} and~\ref{sec:var_sel}.

\subsection{Regression models}\label{sec:reg}
While BAL offers flexibility in specifying $f_2$ and $f_1$, in view of the application of Section~\ref{section5} we consider continuous payoffs and model them with normal multiple linear regressions, namely
\begin{equation}\label{eq:normal_reg}
\begin{aligned}
f_{2}(y; a_{i1},a_{i2},\bm{z}_{i2},\btheta_2)&=f_{\textsc{N}}(y; \btheta_2^\intercal [\bm{z}_{i2} \times (a_{i1},a_{i2})],\sigma_2^2),\\
     f_{1}(y; a_{i1},\bm{z}_{i1},\btheta_1)&=f_{\textsc{N}}(y; \btheta_1^\intercal [\bm{z}_{i1} \times a_{i1}],\sigma_1^2),
    \end{aligned}
\end{equation}
where $f_{\textsc{N}}(y;\mu,\sigma^2)$ is the density 
of a normal random variable with mean $\mu$ and variance $\sigma^2$, and the notation $[\bm{z} \times \bm{a}]$ indicates the vector of regressors including $\bm{z}$, $\bm{a}$ and their interaction terms. Given the continuity of the payoffs, the equivalences in \eqref{eq:alt_opt1} can be rewritten as
\begin{equation}\label{eq:alt_opt3}
    a_{i2}^\opt=a_{i2} \iff \Pr(Y_{i2}\geq\bar{Y}_{i2})>0.5,\quad  
    a_{i1}^\opt=a_{i1} \iff \Pr(Y_{i}^\opt\geq\bar{Y}_{i}^\opt)>0.5.
\end{equation}
The next proposition formalizes the validity of \eqref{eq:alt_opt1} and thus of \eqref{eq:alt_opt3}. 
\begin{proposition}\label{proposition1}
If $f_2$ and $f_1$ are normal linear regressions, specified as in \eqref{eq:normal_reg}, then the equivalences in \eqref{eq:alt_opt1} hold.
\end{proposition}
While normal linear regression is a convenient and widely used choice in the DTR literature \citep[see, e.g.,][]{chakraborty2014dynamic}, we note that the applicability of the BAL approach is broader, as the equivalences in \eqref{eq:alt_opt1} also hold for other specifications of the regression functions $f_2$ and $f_1$, e.g. logistic and Poisson, which are 
relevant in a variety of applications. More details are provided in Appendix \ref{sec:appA}.

\subsection{Multi-stage variable selection}
\label{sec:var_sel}
For the sake of compact notation, we introduce $\tilde{\bm{z}}_{ij}$ to denote the $d_j$-dimensional vector of regressors at the $j$th stage for the $i$th individual, including both the covariates $\bm{z}_{ij}$ and the corresponding interactions with current and past treatments. The focus of our work is on data, such as the MIMIC-III data, 
where the number of individual covariates at each stage is large. This emphasizes the need for dimensionality reduction, especially when interaction terms are incorporated into the regressions, as in \eqref{eq:normal_reg}. We propose to achieve this by defining a joint spike-and-slab prior for the regression coefficients $\btheta_1$ and $\btheta_2$, which we name \emph{dependent spike-and-slab} (DSS). Our proposal takes into account that some regressors, and similarly some interactions between regressors and stage-specific actions, might appear in both $\tilde{\bm{z}}_{i1}$ and $\tilde{\bm{z}}_{i2}$. Without loss of generality, we assume the common elements are in the first $d\leq\min( d_1, d_2)$ positions of $\tilde{\bm{z}}_{i1}$ and $\tilde{\bm{z}}_{i2}$. Although their values may vary across stages, we adopt a joint model for the associated regression coefficients, based on the expectation that the same covariate may play similar roles within the two regression models. To achieve this, we design a prior that induces positive correlation between the inclusion variables for the same regressor across both stages. The idea of introducing dependence among spike-and-slab priors was recently explored by \citet{Zha23}; however, our formulation is novel in specifically addressing the multi-stage structure of DTRs. Among the various spike-and-slab prior specifications discussed in the literature \citep[see][]{Tad21}, we opt for absolutely continuous spikes and model both spike and slab components as normal distributions, with the spike variance set to a small fraction $r$ of the slab variance. 
We note, however, that the structure of the DSS prior can be extended to alternative spike-and-slab specifications, our choice being driven primarily by the computational convenience of the absolutely continuous formulation.
We let $\delta_{lj}$, with $l=1,\ldots,d_j$, be an indicator variable such that $\delta_{lj}=1$ ($\delta_{lj}=0$) if the $l$th regressor is (is not) included in the $j$th regression, and call $w_{lj}=\Pr(\delta_{lj}=1)$. We assume the prior variance of the regression coefficient $\theta_{lj}$ is equal to $\psi_{lj}$ if $\delta_{lj}=1$, and equal to $r\psi_{lj}$ if $\delta_{lj}=0$. We introduce the notation $\bdelta_{j}=(\delta_{1j},\ldots,\delta_{d_j j})$, $\bpsi_{j}=(\psi_{1j},\cdots,\psi_{d_j j})$ and $\bw_j=(w_{1j},\ldots,w_{pj})$, and define the DSS prior $p(\btheta_1,\btheta_2)$ as the distribution arising from the following hierarchy: 
\begin{align}\label{eq:ss_prior1}
    (\btheta_{1},\btheta_{2})\mid \bdelta_{1},\bdelta_{2},\bpsi_{1},\bpsi_{2} &\sim\prod_{j=1}^2\prod_{l_j=1}^{d_j}\textsc{N}\left(0,r^{1-\delta_{l_jj}}\psi_{l_jj}\right)&\notag\\(\bpsi_{1},\bpsi_{2})&\sim\prod_{j=1}^2\prod_{l_j=1}^{d_j}\textsc{inv-Gamma}(\nu,Q)\notag\\
    (\bdelta_{1},\bdelta_{2})\mid \bw_1, \bw_2 &\sim \prod_{j=1}^2\prod_{l_j=1}^{d_j}\textsc{Bern}(w_{l_jj})&\\
    w_{l1}=w_{l2}&\simiid \textsc{Beta}(a,b) & l=1,\dots,d \notag\\
    w_{lj}&\simiid \textsc{Beta}(a,b) & l_j=d+1,\dots,d_j, \;j=1,2.\notag
 \end{align}
The distribution in \eqref{eq:ss_prior1} is an extension to the case of two vectors of regression coefficients of the normal mixture of inverse gamma spike-and-slab prior proposed by \citet[][]{Ish05}. The indicator variables of common regressors, that is $(\delta_{l1},\delta_{l2})$ for $l=1,\ldots,d$, are marginally dependent, which 
favors borrowing of information across the two stages when performing variable selection. The DSS prior thus formalizes the idea that there might be positive correlation between the significance of the same regressor across stages. Following \eqref{eq:ss_prior1}, $\delta_{l_1 1}$ and $\delta_{l_2 2}$ are uncorrelated unless $l_1 = l_2 \in \{1, \ldots, d\}$, in which case
\begin{equation}\label{eq:corr} \rho(\delta_{l_11},\delta_{l_22})=
        \frac{1}{1+a+b}.
\end{equation}

\subsection{Posterior computation}\label{section3}
The joint distribution of observations, pseudo-outcomes, and parameters is defined by combining the BAL model in \eqref{eq:aug_model} with the regression functions in \eqref{eq:normal_reg}, the DSS prior in \eqref{eq:ss_prior1}, and independent normal priors for $(\sigma_1^2, \sigma_2^2)$. The full conditional distributions of all model components, derived in Appendix \ref{sec:FC}, 
form the basis of a Gibbs sampling algorithm for posterior inference, with steps outlined in Algorithm~\ref{algo} 
in the Appendix. We note that, while the inclusion of the pseudo-outcomes primarily serves to facilitate posterior sampling, Algorithm~\ref{algo} 
also enables direct sampling from their posterior distribution. As a by-product, the probability $\mathbb{E}[A_{ij}^\opt]$ that the optimal treatment at the $j$th stage for the $i$th individual is equal to 1 can be estimated with the frequency of the event $A^\opt_{ij}=1$, as defined in \eqref{eq:alt_opt1}, in the posterior sample. The corresponding estimated individual optimal treatment $\hat{a}_{ij}^\opt$ is given by the mode of the posterior distribution of $A_{ij}^\opt$.

\section{Simulation study}\label{section4}
We conducted an extensive simulation study to assess the performance of the BAL model. We compared two of its variants---BAL with normal regression functions \eqref{eq:normal_reg} and the DSS prior \eqref{eq:ss_prior1} (variant referred to as DSS), and with independent spike-and-slab priors (variant referred to as ISS)---against three benchmark methods commonly used in the frequentist literature. Namely, i)  Q-learning with lasso penalty \citep{qian2011performance}; ii) Q-learning with random forests (RF) \citep{min2022q}; iii) augmented outcome-weighted learning (AOWL) with lasso penalty \citep{zhao2012estimating,liu2018augmented}. The latter is not designed for variable selection and was used solely to estimate optimal treatment regimes. \\  
The DSS prior is specified by setting $d = d_1$, implying that the number of common regressors across stages is equal to the number of regressors in stage 1. The ISS specification is obtained by omitting common regressors across stages in the DSS prior, i.e., by setting $d = 0$ in \eqref{eq:ss_prior1}. For both DSS and ISS, we set $\nu = 3$, $Q = 4$, and assign independent inverse gamma priors to $a$ and $b$, with shape and scale parameters equal to 1. The ratio $r$ is set to $0.001$, and a sensitivity analysis in Appendix \ref{sec:sensitivity} suggests that BAL is robust to this choice.
Throughout this study, data are analyzed using DSS and ISS by running Algorithm \ref{algo} 
for 10,000 iterations, with the first 5,000 discarded as burn-in. As shown in Figure \ref{fig_traceplot} 
in the Appendix, standard diagnostic tools were applied and revealed no evidence of convergence or mixing issues. Implementation details for the alternative methods considered are provided in Appendix~\ref{sec:alt_meth}. 
Our investigation focuses on the models' ability to accurately identify significant regressors across stages, and on the effect of variable selection on identifying the optimal DTR. The measures used to evaluate these aspects are detailed in Appendix \ref{sec:measures}. 
The simulation study consists of three experiments: Section~\ref{PART1} explores scenarios with different numbers of stage-specific regressors, and varying degrees of correlation between the significance of common regressors across stages; Section~\ref{PART2} focuses on the analysis of simulated datasets characterized by varying number of treatments, thus extending the framework described in Section~\ref{section2} to sets $\mathcal{A}_1$ and $\mathcal{A}_2$ with cardinality larger than two. A third experiment, focusing on data generated via a mechanism that goes beyond linearity, is presented in Appendix~\ref{sec:nonlin}.

\subsection{First Experiment}
\label{PART1}
We consider data generated according to the mechanism detailed in Appendix \ref{sec:dgp1}, 
an adaptation of the class of generative models proposed by \cite{chakraborty2010inference}.  The simulation study includes nine simulation scenarios, defined by the combination of three values for the number of main effects, $k\in\{10,20,30\}$ (corresponding to $d_1\in\{22,42,62\}$ and $d_2\in \{33,63,93\}$, when also counting treatment effects), and three levels of correlation $\rho^*=\rho(\delta_{l1}^*,\delta_{l2}^*)$, for $l\in\{1,\ldots,d_1\}$: weak ($0.3$), medium ($0.6$) and strong ($0.9$). The variables $\delta_{i1}^*$ and $\delta_{i2}^*$ indicate whether, under the data generating process, the $l$th predictor is significant at the $j$th stage ($\delta_{lj}^* = 1$) or not ($\delta_{lj}^* = 0$). We consider datasets simulated under two regimes: $n$ smaller than $d_2$ and $n$ larger than $d_2$. The sample sizes are $n \in \{25, 50\}$ for $k=10$, $n \in \{50, 100\}$ for $k=20$, and $n \in \{75, 150\}$ for $k=30$. Larger sample sizes meet the condition $n > d_2 > d_1$, while smaller samples are more challenging because $d_1 < n < d_2$, resulting in more regressors than observations at stage 2. For each scenario, 100 datasets are generated independently.

\begin{table}[h]
\caption{First experiment: simulated data with $n$ smaller than $d_2$. For stage 2 and stage 1, three summaries of the accuracy in selecting significant variables (proportion of FN, proportion of FP, and F$_1$ score), for DSS, ISS, Q-learning with lasso (QL), and Q-learning with RF (QRF), across the nine scenarios obtained by specifying $(k,n)\in\{(10,25),(20,50),(30,75)\}$  and $\rho^*\in\{0.3,0.6,0.9\}$. Results are based on 100 replicated datasets.}
\centering
\resizebox{1\textwidth}{!}{
\begin{tabular}{ccccccccccccccccccc}
\Hline
\multicolumn{18}{c}{Stage 2}                                                                                                                         \\ \hline
\multirow{2}{*}{$n$} & \multirow{2}{*}{$d_2$} & \multirow{2}{*}{$\rho^*$} &  & \multicolumn{4}{c}{FN} &  & \multicolumn{4}{c}{FP} &  & \multicolumn{4}{c}{F$_1$ score} \\ \cline{5-8} \cline{10-13} \cline{15-18} 
 &  &  & & DSS    & ISS   & QL  &QRF   &  & DSS    & ISS   & QL  &QRF   &  & DSS      & ISS     & QL   &QRF   \\ \hline 
 \multirow{3}{*}{25}            & \multirow{3}{*}{33}
& 0.3 &  &0.158 & 0.196 & 0.430 &0.467  & & 0.247 & 0.178 & 0.163 & 0.211  &  & 0.704 & 0.733 & 0.532 & 0.512  \\
   &    & 0.6 &  &0.123 & 0.182 & 0.382 &0.457  & & 0.217 & 0.160 & 0.181 & 0.203  &  & 0.744 & 0.753 & 0.564 & 0.513  \\
   &    & 0.9 &  &0.131 & 0.204 & 0.386 &0.469  & & 0.180 & 0.154 & 0.172 & 0.217  &  & 0.770 & 0.740 & 0.567 & 0.500  \\[5pt]
 \multirow{3}{*}{50}            & \multirow{3}{*}{63} & 0.3 &  &0.039 & 0.060 & 0.233 &0.482  & & 0.237 & 0.154 & 0.206 & 0.216  &  & 0.788 & 0.841 & 0.666 & 0.501  \\
   &    & 0.6 &  &0.038 & 0.062 & 0.242 &0.489  & & 0.192 & 0.141 & 0.207 & 0.215  &  & 0.820 & 0.847 & 0.659 & 0.499  \\
   &    & 0.9 &  &0.025 & 0.049 & 0.239 &0.499  & & 0.169 & 0.178 & 0.204 & 0.219  &  & 0.859 & 0.836 & 0.666 & 0.491  \\[5pt]
 \multirow{3}{*}{75}            & \multirow{3}{*}{93} & 0.3 &  &0.016 & 0.019 & 0.137 &0.514 & & 0.185 & 0.098 & 0.241 & 0.223  &  & 0.832 & 0.904 & 0.706 & 0.481  \\
   &    & 0.6 &  &0.010 & 0.022 & 0.178 &0.498  & & 0.151 & 0.123 & 0.225 & 0.218  &  & 0.866 & 0.884 & 0.691 & 0.494  \\
   &    & 0.9 &  &0.012 & 0.018 & 0.192 &0.492  & & 0.098 & 0.130 & 0.212 & 0.217  &  & 0.915 & 0.888 & 0.690 & 0.497 \\ \hline
\multicolumn{18}{c}{Stage 1}                                                                                                                               \\ \hline
\multirow{2}{*}{$n$} & \multirow{2}{*}{$d_1$} & \multirow{2}{*}{$\rho^*$} &  & \multicolumn{4}{c}{FN} &  & \multicolumn{4}{c}{FP} &  & \multicolumn{4}{c}{F$_1$ score} \\ \cline{5-8} \cline{10-13} \cline{15-18} 
 &  &  & & DSS    & ISS   & QL  &QRF   &  & DSS    & ISS   & QL  &QRF   &  & DSS      & ISS     & QL   &QRF       \\ \hline
 \multirow{3}{*}{25}            & \multirow{3}{*}{22} & 0.3 &  & 0.192 & 0.348 & 0.395 &0.457  & & 0.327 & 0.137 & 0.190 & 0.231 &  & 0.624 & 0.644 & 0.541 & 0.496 \\
   &    & 0.6 &  & 0.140 & 0.294 & 0.378 &0.409 & & 0.244 & 0.113 & 0.190 & 0.203  &  & 0.719 & 0.709 & 0.570 & 0.551 \\
   &    & 0.9 &  & 0.085 & 0.295 & 0.327 &0.379  & & 0.168 & 0.115 & 0.185 & 0.207  &  & 0.806 & 0.693 & 0.597 & 0.555  \\[5pt]
 \multirow{3}{*}{50}            & \multirow{3}{*}{42} & 0.3 &  & 0.072 & 0.091 & 0.216 &0.462  & & 0.322 & 0.164 & 0.227 & 0.218  &  & 0.705 & 0.819 & 0.667 & 0.513  \\
   &    & 0.6 &  & 0.045 & 0.106 & 0.245 &0.437  & & 0.218 & 0.151 & 0.184 & 0.207  &  & 0.792 & 0.821 & 0.681 & 0.539  \\
   &    & 0.9 &  & 0.022 & 0.079 & 0.214 &0.444 & & 0.158 & 0.198 & 0.205 & 0.208  &  & 0.870 & 0.822 & 0.686 & 0.533  \\[5pt]
 \multirow{3}{*}{75}            & \multirow{3}{*}{62} & 0.3 &  & 0.033 & 0.040 & 0.123 &0.464  & & 0.271 & 0.087 & 0.232 & 0.211 &  & 0.754 & 0.910 & 0.730 & 0.521  \\
   &    & 0.6 &  & 0.013 & 0.036 & 0.163 &0.457  & & 0.183 & 0.136 & 0.204 & 0.208  &  & 0.838 & 0.881 & 0.719 & 0.527  \\
   &    & 0.9 &  & 0.007 & 0.034 & 0.156 &0.444 & & 0.093 & 0.139 & 0.213 & 0.206 &  & 0.920 & 0.887 & 0.712 & 0.533  \\ \hline
\end{tabular}
}
\label{TABLE first experiment_1}
\end{table}
Table \ref{TABLE first experiment_1} summarizes the performance of the considered methods, excluding AOWL, in identifying significant variables for data simulated under the regime where $n < d_2$. The F$_1$ scores clearly indicate that the two BAL variants outperform all competing approaches at both stages and across all scenarios. In particular, DSS tends to perform better when the data are generated under strong correlation, whereas ISS is preferable when the correlation is weak. Among the alternative methods, Q-learning with lasso consistently outperforms Q-learning with random forests. The latter appears to struggle with the moderate sample sizes considered in this study. These differences across the five approaches are further reflected in the false negative (FN) and false positive (FP) rates, with DSS achieving the lowest FN rate in most scenarios, and ISS attaining the lowest FP rate.
\begin{table}[H]
\centering
\caption{First experiment: simulated data with $n$ smaller than $d_2$. For stage 2, stage 1 and overall, two measures of prediction accuracy (MAE and ER), for DSS, ISS, Q-learning with lasso (QL), Q-learning with RF (QRF) and AOWL, across the nine scenarios obtained by specifying $(k,n)\in\{(10,25),(20,50),(30,75)\}$ and $\rho^*\in\{0.3,0.6,0.9\}$. Results are based on 100 replicated datasets.}
\resizebox{1\textwidth}{!}{
\begin{tabular}{ccccccccccccccc}
\Hline
\multicolumn{15}{c}{Stage 2}                                                                                \\ \hline
         \multirow{2}{*}{$n$} &   \multirow{2}{*}{$d_2$}                               &   \multirow{2}{*}{$\rho^*$}   & & \multicolumn{5}{c}{MAE}                                            &                              & \multicolumn{5}{c}{ER}                                                           \\ \cline{5-9} \cline{11-15} 
 &   & & & DSS                             & ISS                           & QL     &QRF & AOWL  &                    & DSS                           & ISS                           & QL    &QRF & AOWL                        \\ \hline
\multicolumn{1}{c}{}                                   & \multicolumn{1}{c}{}                            & 0.3  & & 0.248 & 0.274 & 0.648 & 0.935 & 0.911 & & 0.136 & 0.144 & 0.219 & 0.303 & 0.288     \\
\multicolumn{1}{c}{}                                   & \multicolumn{1}{c}{}  & 0.6 & & 0.260 & 0.264 & 0.591 & 0.881 & 0.892 & & 0.130 & 0.130 & 0.219 & 0.296 & 0.284 \\
\multicolumn{1}{c}{\multirow{-3}{*}{25}}               & \multicolumn{1}{c}{\multirow{-3}{*}{33}}     & 0.9 & & 0.307 & 0.339 & 0.790 & 0.909 & 0.882 & & 0.148 & 0.160 & 0.264 & 0.302 & 0.282 \\[5pt]
\multicolumn{1}{c}{}                                   & \multicolumn{1}{c}{} & 0.3 & & 0.151 & 0.145 & 0.459 & 1.294 & 1.174 & & 0.087 & 0.084 & 0.162 & 0.302 & 0.280 \\
\multicolumn{1}{c}{}                                   & \multicolumn{1}{c}{}   & 0.6 & & 0.179 & 0.191 & 0.499 & 1.279 & 1.216 & & 0.079 & 0.083 & 0.152 & 0.285 & 0.283 \\
\multicolumn{1}{c}{\multirow{-3}{*}{50}}               & \multicolumn{1}{c}{\multirow{-3}{*}{63}}  & 0.9 & & 0.134 & 0.168 & 0.500 & 1.346 & 1.178 & & 0.076 & 0.086 & 0.162 & 0.298 & 0.279\\ [5pt]
\multicolumn{1}{c}{}                                   & \multicolumn{1}{c}{}  & 0.3 & & 0.084 & 0.079 & 0.325 & 1.666 & 1.317 & & 0.055 & 0.053 & 0.114 & 0.300 & 0.269 \\
\multicolumn{1}{c}{}                                   & \multicolumn{1}{c}{}  & 0.6 & & 0.064 & 0.099 & 0.501 & 1.613 & 1.310 & & 0.050 & 0.058 & 0.135 & 0.297 & 0.261 \\
\multicolumn{1}{c}{\multirow{-3}{*}{75}}               & \multicolumn{1}{c}{\multirow{-3}{*}{93}} & 0.9 & & 0.069 & 0.114 & 0.466 & 1.633 & 1.329 & & 0.049 & 0.060 & 0.144 & 0.306 & 0.272 \\   \hline
\multicolumn{15}{c}{Stage 1} 
\\
\hline
\multirow{2}{*}{$n$} & \multirow{2}{*}{$d_1$} & \multirow{2}{*}{$\rho^*$} & & \multicolumn{5}{c}{MAE}                                       &                                   & \multicolumn{5}{c}{ER}                                                                        \\ \cline{5-9} \cline{11-15} 
 &   & & & DSS                             & ISS                           & QL     &QRF & AOWL  &                    & DSS                           & ISS                           & QL    &QRF & AOWL                                            \\ \hline
 \multicolumn{1}{c}{} & \multicolumn{1}{c}{} & 0.3 & & 0.463 & 0.495 & 0.823 & 0.924 & 1.043 & & 0.183 & 0.190 & 0.267 & 0.292 & 0.316  \\
\multicolumn{1}{c}{} & \multicolumn{1}{c}{} & 0.6 & & 0.386 & 0.441 & 0.794 & 0.865 & 0.911 & & 0.179 & 0.196 & 0.256 & 0.271 & 0.284  \\
 \multicolumn{1}{c}{\multirow{-3}{*}{25}} & \multicolumn{1}{c}{\multirow{-3}{*}{22}}  & 0.9 & & 0.332 & 0.404 & 0.633 & 0.840 & 0.906 & & 0.175 & 0.193 & 0.239 & 0.293 & 0.294 \\ [5pt]
 \multicolumn{1}{c}{} & \multicolumn{1}{c}{} & 0.3 & & 0.324 & 0.362 & 0.574 & 1.312 & 1.231 & & 0.129 & 0.131 & 0.186 & 0.298 & 0.290  \\
\multicolumn{1}{c}{} & \multicolumn{1}{c}{} & 0.6 & & 0.266 & 0.333 & 0.617 & 1.300 & 1.216 & & 0.112 & 0.131 & 0.185 & 0.296 & 0.293  \\
 \multicolumn{1}{c}{\multirow{-3}{*}{50}} & \multicolumn{1}{c}{\multirow{-3}{*}{42}}  & 0.9 & & 0.214 & 0.308 & 0.611 & 1.321 & 1.127 & & 0.094 & 0.122 & 0.176 & 0.291 & 0.270  \\ [5pt]
 \multicolumn{1}{c}{} & \multicolumn{1}{c}{} & 0.3 & & 0.185 & 0.190 & 0.416 & 1.783 & 1.371 & & 0.088 & 0.086 & 0.136 & 0.332 & 0.284  \\
\multicolumn{1}{c}{} & \multicolumn{1}{c}{} & 0.6 & & 0.136 & 0.199 & 0.571 & 1.634 & 1.383 & & 0.071 & 0.084 & 0.151 & 0.294 & 0.274  \\
 \multicolumn{1}{c}{\multirow{-3}{*}{75}} & \multicolumn{1}{c}{\multirow{-3}{*}{62}}  & 0.9 & & 0.119 & 0.191 & 0.469 & 1.575 & 1.340 & & 0.063 & 0.082 & 0.142 & 0.301 & 0.276 \\ \hline
\multicolumn{15}{c}{Overall}                                                             \\ \hline
\multirow{2}{*}{$n$} & \multirow{2}{*}{$(d_1,d_2)$} & \multirow{2}{*}{$\rho^*$} & &  \multicolumn{5}{c}{MAE}                                        &                                  & \multicolumn{5}{c}{ER}                                                                        \\ \cline{5-9} \cline{11-15}                                          &   & & & DSS                             & ISS                           & QL     &QRF & AOWL  &                    & DSS                           & ISS                           & QL    &QRF & AOWL                                        \\ \hline
\multicolumn{1}{c}{}                                   & \multicolumn{1}{c}{}                         & 0.3 &      & 0.711 & 0.770 & 1.472 &1.860&1.954	&& 0.290 & 0.299 & 0.418 & 0.501&0.507                \\   
\multicolumn{1}{c}{}                                   & \multicolumn{1}{c}{}                            & 0.6 &     & 0.646 & 0.706 & 1.384 &1.746&1.803&& 0.287 & 0.304 & 0.418 &0.476&0.490         \\
\multicolumn{1}{c}{\multirow{-3}{*}{25}}               & \multicolumn{1}{c}{\multirow{-3}{*}{(22, 33)}}    & 0.9 &     & 0.639 & 0.742 & 1.423 &1.750&1.788&& 0.288 & 0.314 & 0.435  &0.504  &0.484           \\[5pt]
\multicolumn{1}{c}{}                                   & \multicolumn{1}{c}{}                            & 0.3 &       & 0.475 & 0.507 & 1.033 &2.610&2.405&& 0.205 & 0.203 & 0.315 &0.508&0.483     \\
\multicolumn{1}{c}{}                                   & \multicolumn{1}{c}{}                            & 0.6 &       & 0.445 & 0.524 & 1.116 & 2.579 & 2.432 & & 0.179 & 0.202 & 0.301 & 0.494 & 0.490     \\
\multicolumn{1}{c}{\multirow{-3}{*}{50}}               & \multicolumn{1}{c}{\multirow{-3}{*}{(42, 63)}}  & 0.9 &        & 0.348 & 0.476 & 1.111 & 2.667 & 2.304 && 0.162 & 0.196 & 0.306 & 0.496 & 0.473             \\[5pt]
\multicolumn{1}{c}{}                                   & \multicolumn{1}{c}{}                            & 0.3 &        & 0.269 & 0.269 & 0.740 & 3.451 & 2.688 & & 0.136 & 0.133 & 0.231 & 0.529 & 0.474         \\
\multicolumn{1}{c}{}                                   & \multicolumn{1}{c}{}                            & 0.6 &          & 0.200 & 0.298 & 1.073 & 3.248 & 2.693 && 0.116 & 0.135 & 0.261 & 0.507&  0.464     \\
\multicolumn{1}{c}{\multirow{-3}{*}{75}}               & \multicolumn{1}{c}{\multirow{-3}{*}{(62, 93)}} & 0.9 &        & 0.189 & 0.305 & 0.934 & 3.208 & 2.669 && 0.109 & 0.135 & 0.260 & 0.514& 0.474  \\ \hline 
\end{tabular}
}
\label{TABLE first experiment_2}
\end{table}
Table~\ref{TABLE first experiment_2} summarizes the empirical performance of the five methods by comparing the estimated optimal treatment regimes to the true ones, under the regime where $n < d_2$. Across all scenarios and stages, the two BAL variants outperform the competing methods in terms of treatment error rate (ER) and payoff mean absolute error (MAE). DSS performs as well as or better than ISS, with a more marked advantage in settings characterized by strong correlation $\rho^*$. Among the alternative approaches, Q-learning with RF and AOWL consistently exhibit the weakest performance. Overall performance measures, which aggregate results across stages, reinforce these findings. The results of a similar analysis, this time based on data simulated under the regime where $n > d_2$, are reported in Tables~\ref{TABLE first experiment_1_1} 
and \ref{TABLE first experiment_2_2} 
in the Appendix. A comparison of the two regimes highlights differences in method performance. 
When $n > d_2$, DSS and ISS exhibit similar behavior, particularly with respect to ER and MAE. This suggests that the advantages of the DSS prior over ISS are more pronounced when the sample size is small relative to the number of regressors. Q-learning with RF and AOWL continue to underperform, whereas Q-learning with lasso performs only slightly worse than the two BAL variants, particularly in terms of ER and MAE. 

\subsection{Second experiment}\label{PART2}
In the second experiment, we set $\mathcal{A}_1 = \mathcal{A}_2 = \{0, 1, \ldots, T - 1\}$ and consider four simulation scenarios, defined by varying the number of treatments, $T \in \{4, 8\}$, and the sample size, $n \in \{200, 400\}$. Here, the number of individual covariates is fixed at $k = 10$, and data are generated with $\rho^* = 0.6$ using a data-generating process similar to that of the first experiment, as described in Appendix \ref{sec:dgp2}. 
Given that the number of main effects is $k = 10$ and that $d_1=T(k+1)$ and $d_2=(2T-1)(k+1)$, we have $(d_1, d_2) = (44, 77)$ when $T = 4$ and $(d_1, d_2) = (88,165)$ when $T = 8$. All scenarios considered in this experiment fall within the regime where $n > d_2$. For each scenario, we independently generate 100 datasets. AOWL is excluded from this second experiment because the \texttt{DTRlearn2} R package only supports data with binary treatments.
\begin{table}[h]
\centering
\caption{Second experiment. For stage 2 and stage 1, three summaries of the accuracy in selecting significant variables (proportion of FN, proportion of FP, and F$_1$ score), for DSS, ISS, Q-learning with lasso (QL), and Q-learning with RF (QRF), across the four scenarios obtained by specifying $n\in\{200,400\}$ and $T\in\{4,8\}$. The number of individual covariates is $k=10$, while $\rho^*=0.6$. Results are based on 100 replicated datasets.}
\resizebox{1\textwidth}{!}{ 
\begin{tabular}{cccccccccccccccccc}
\Hline
\multicolumn{18}{c}{Stage 2} \\ \hline
\multirow{2}{*}{$T$} & \multirow{2}{*}{$d_2$} & \multirow{2}{*}{$n$} &  & \multicolumn{4}{c}{FN} &  & \multicolumn{4}{c}{FP} &  & \multicolumn{4}{c}{F$_1$ score} \\ \cline{5-8} \cline{10-13} \cline{15-18} 
 &  &  & & DSS & ISS & QL & QRF &   & DSS & ISS & QL & QRF &  & DSS & ISS & QL & QRF  \\ \hline
\multirow{2}{*}{4}   & \multirow{2}{*}{99}    & 200 &  & 0.003 & 0.004 & 0.015 & 0.344  &  & 0.025 & 0.018 & 0.124 & 0.171  &  & 0.972 & 0.978 & 0.873 & 0.626  \\
                     &                        & 400 &  & 0.002 & 0.003 & 0.007 & 0.306  &  & 0.009 & 0.005 & 0.085 & 0.151  &  & 0.989 & 0.993 & 0.913 & 0.672  \\[5pt]
\multirow{2}{*}{8}   & \multirow{2}{*}{187}   & 200 &  & 0.006 & 0.010 & 0.023 & 0.467 &  & 0.084 & 0.056 & 0.244 & 0.203 &  & 0.914 & 0.938 & 0.780 & 0.533  \\
                     &                        & 400 &  & 0.003 & 0.004 & 0.010 & 0.382 &  & 0.026 & 0.021 & 0.102 & 0.168  &  & 0.970 & 0.975 & 0.895 & 0.614 \\ \hline
\multicolumn{18}{c}{Stage 1} \\ \hline
\multirow{2}{*}{$T$} & \multirow{2}{*}{$d_1$} & \multirow{2}{*}{$n$} &  & \multicolumn{4}{c}{FN} &  & \multicolumn{4}{c}{FP} &  & \multicolumn{4}{c}{F$_1$ score} \\ \cline{5-8} \cline{10-13} \cline{15-18} 
 &  &  & & DSS & ISS & QL & QRF &   & DSS & ISS & QL & QRF &  & DSS & ISS & QL & QRF \\ \hline
\multirow{2}{*}{4}   & \multirow{2}{*}{55}    & 200 &  & 0.004 & 0.008 & 0.017 & 0.334 &  & 0.049 & 0.002 & 0.093 & 0.173 &  & 0.948 & 0.993 & 0.899 & 0.634 \\
                     &                        & 400 &  & 0.002 & 0.003 & 0.007 & 0.291  &  & 0.020 & 0.003 & 0.077 & 0.151 &  & 0.977 & 0.995 & 0.922 & 0.677  \\[5pt]
\multirow{2}{*}{8}   & \multirow{2}{*}{99}    & 200 &  & 0.009 & 0.012 & 0.036 & 0.425  &  & 0.141 & 0.029 & 0.160 & 0.195 &  & 0.854 & 0.964 & 0.835 & 0.558\\
                     &                        & 400 &  & 0.004 & 0.007 & 0.017 & 0.364  &  & 0.054 & 0.003 & 0.060 & 0.165  &  & 0.943 & 0.994 & 0.932 & 0.626 \\ \hline
\end{tabular}
}
\label{TABLE second_1}
\end{table}
When analyzing datasets with more than two treatments, the advantages of two two BAL variants are confirmed, compared to both versions of Q-learning.
Table \ref{TABLE second_1} summarizes the performance of the four methods in selecting significant regressors. DSS and ISS perform similarly and outperform the other methods, with DSS exhibiting lower FN rates and ISS lower FP rates across all scenarios. As the sample size increases, all methods show improved performance, reflected in higher F$_1$ scores and reduced FN and FP rates.
\begin{table}[h!]
\centering
\caption{Second experiment. For stage 2, stage 1 and overall, two measures of prediction accuracy (MAE and ER), for DSS, ISS, Q-learning with lasso (QL), and Q-learning with RF (QRF), across the four scenarios obtained by specifying $n\in\{200,400\}$ and $T\in\{4,8\}$. The number of individual covariates is $k=10$, while $\rho^*=0.6$. Results are based on 100 replicated datasets.} 
\begin{tabular}{ccccccccccccc}
\Hline
\multicolumn{13}{c}{Stage 2}\\ 
\hline
\multirow{2}{*}{$T$} & \multirow{2}{*}{$d_2$} & \multirow{2}{*}{$n$} & & \multicolumn{4}{c}{MAE} & & \multicolumn{4}{c}{ER} \\ \cline{5-8} \cline{10-13} 
 & & & & DSS & ISS & QL & QRF  & & DSS & ISS & QL & QRF  \\ \hline
\multirow{2}{*}{4} & \multirow{2}{*}{99} & 200 & & 0.011 & 0.010 & 0.015 & 1.050  & & 0.042 & 0.043 & 0.051 & 0.330  \\
  &    & 400 & & 0.004 & 0.004 & 0.007 & 0.680  & & 0.034 & 0.034 & 0.035 & 0.270 \\[5pt]
\multirow{2}{*}{8} & \multirow{2}{*}{187} & 200 & & 0.055 & 0.053 & 0.089 & 2.933 & & 0.084 & 0.083 & 0.108 & 0.556  \\
  &     & 400 & & 0.013 & 0.013 & 0.015 & 1.994  & & 0.046 & 0.045 & 0.048 & 0.485  \\ \hline
\multicolumn{13}{c}{Stage 1}\\ \hline
\multirow{2}{*}{$T$} & \multirow{2}{*}{$d_1$} & \multirow{2}{*}{$n$} & & \multicolumn{4}{c}{MAE} & & \multicolumn{4}{c}{ER} \\ \cline{5-8} \cline{10-13} 
 & & & & DSS & ISS & QL & QRF  & & DSS & ISS & QL & QRF \\ \hline
\multirow{2}{*}{4} & \multirow{2}{*}{55} & 200 & & 0.014 & 0.012 & 0.024 & 0.968  & & 0.048 & 0.047 & 0.056 & 0.332  \\
  &    & 400 & & 0.006 & 0.005 & 0.006 & 0.611  & & 0.039 & 0.038 & 0.029 & 0.260  \\[5pt]
\multirow{2}{*}{8} & \multirow{2}{*}{99} & 200 & & 0.084 & 0.091 & 0.140 & 2.562  & & 0.105 & 0.107 & 0.128 & 0.550 \\
  &    & 400 & & 0.022 & 0.022 & 0.026 & 1.909 & & 0.057 & 0.057 & 0.059 & 0.473  \\ \hline
\multicolumn{13}{c}{Overall}\\\hline
\multirow{2}{*}{$T$} & \multirow{2}{*}{$(d_1, d_2)$} & \multirow{2}{*}{$n$} & & \multicolumn{4}{c}{MAE} & & \multicolumn{4}{c}{ER} \\ \cline{5-8} \cline{10-13} 
 & & & & DSS & ISS & QL & QRF  & & DSS & ISS & QL & QRF  \\ \hline
\multirow{2}{*}{4} & \multirow{2}{*}{(55, 99)} & 200 & & 0.024 & 0.023 & 0.038 & 2.018  & & 0.089 & 0.088 & 0.104 & 0.553 \\
  &          & 400 & & 0.010 & 0.010 & 0.013 & 1.291  & & 0.071 & 0.071 & 0.063 & 0.456  \\[5pt]
\multirow{2}{*}{8} & \multirow{2}{*}{(99, 187)} & 200 & & 0.138 & 0.144 & 0.229 & 5.495 & & 0.181 & 0.181 & 0.223 & 0.801 \\
  &           & 400 & & 0.035 & 0.034 & 0.042 & 3.903  & & 0.100 & 0.099 & 0.104 & 0.726 \\ \hline
\end{tabular}
\label{TABLE second_2}
\end{table} 
As shown in Table~\ref{TABLE second_2}, the two BAL variants yield the most accurate predictions, achieving lower MAE and ER values than both Q-learning versions across all scenarios. DSS and ISS exhibit similar performance, indicating that the differences between the two approaches become less pronounced at the larger sample sizes considered in this experiment. Increasing the sample size from 200 to 400 improves the predictive performance of all methods.\\
As an additional evaluation of our method, we examined a scenario with imbalanced treatment assignment by making a minor modification to the data-generating process used in the second simulation experiment. As detailed in Appendix \ref{sec:appPOS}, 
the results indicate that the predictive accuracy of the BAL approach is affected by treatment imbalance, while its impact on variable selection appears less systematic. 

\section{Clinical case study}\label{section5}
We analyze a clinical case study that focuses on hospital care medications for patients in intensive care units (ICU) with severe acute arterial hypertension. As outlined by \cite{zhou2022optimal}, data are extracted from the MIMIC-III clinical database using free-text extraction techniques to retrieve relevant information from unstructured medical records.
The estimation of DTRs is of particular interest because personalized antihypertensive treatments have been shown to result in more effective interventions to manage hypertension \citep{kotchen2016ushering}. 
The data we analyze include patients who (i) stayed in the ICU for at least 3 days, (ii) had a first-day maximum systolic blood pressure (SBP) over 180 mmHg, and (iii) received treatment combinations assigned to at least 1\% of the population in both stages. Stage 1 comprised 742 patients admitted for hypertensive crisis; of these, 320 progressed to stage 2 due to persistent hypertension requiring further medication on the second day. Four classes of antihypertensive agents are commonly used to treat hypertension: angiotensin-converting enzyme inhibitors (ACEi), beta-blockers, calcium channel blockers (CCB), and diuretics. Moreover, combination therapies are known to often yield better outcomes than single-drug treatments \citep{wald2009combination,kotchen2016ushering}. Our goal is to identify the optimal two-stage dynamic DTR for antihypertensive combinations based on patients’ SBP management. Among all possible combinations of the aforementioned antihypertensive agents, the analyzed dataset includes only $T=8$ treatment options: four individual agents and four combinations. Their frequencies are detailed in Table~\ref{fren analysis} 
in the Appendix.
A total of $k=12$ clinical factors are considered across both stages: daily maximum SBP, heart rate (HR), oxygen saturation (SpO2), age, weight, gender, race, smoking status, kidney disease, diabetes, chronic obstructive pulmonary disease (COPD), and chronic hypertension (CH). Notably, potential confounders such as SBP and CH are included, consistent with standard confounder adjustment methods \citep[see][]{chakraborty2013statistical}.
Our DTR approach aims to recommend the optimal combination of antihypertensive agents at stage 1, based on patients’ clinical characteristics immediately after ICU admission. For those with maximum SBP exceeding 140 mmHg the next day, it further recommends treatment adjustments at stage 2, considering both the patient’s clinical factors that day and the initial therapy. We use the decrease in maximum SBP after treatment as the response variable. Specifically, for the $i$th subject, we define $Y_i = Y_{i1} + Y_{i2}$, where $Y_{ij}$ denotes the decrease in maximum SBP after the $j$th stage, and $Y_i$ the overall decrease. To accommodate differing patient counts across stages, we introduce the variable $\mathcal{I}_i$, set to 1 if the $i$th patient received treatment on both days, and 0 if treatment was given only on the first day.
 The stage 2 augmented model is defined solely for patients with $\mathcal{I}_i = 1$ and is specified as follows:
\begin{equation*}
\begin{aligned}
               (Y_{i2},V_{i2}^{(0:7)})
        \mid a_{i1}, a_{i2},\bm{z}_{i2},\btheta_2 &\simind f_{\textsc{N}}(y; \btheta_2^\intercal [\bm{z}_{i2}  \times 
(\bm{a}^*_{i1},\bm{a}^*_{i2})],\sigma_2^2)\prod_{t=0}^{7} f_{\textsc{N}}(v_t; \btheta_2^\intercal [\bm{z}_{i2} \times (\bm{a}^*_{i1},\bm{i}_t)],\sigma_2^2),
\end{aligned}
\end{equation*}
where $\bm{a}^*_{ij} = (a_{ij1}^*, \ldots, a_{ij7}^*) \in \mathcal{D}_{7}$ are vectors of dummy variables such that $a_{ijt}^* = \mathbb{I}_{\{t\}}(a_{ij})$, and  $\bm{i}_t \in \mathcal{D}_{7}$ is the null vector if $t = 0$, and is defined as the vector of zeros with a 1 in position $t$, if $t \in \{1, \ldots, 7\}$. Without loss of generality, we set $V_{i2}^{(t)} = 0$ for every $t \in \{0, \ldots, 7\}$ when $\mathcal{I}_i = 0$, which is equivalent to assuming that if the decrease in maximum SBP after stage 1 was effective, no additional decrease is considered at stage 2. In this case, the stage 1 treatment is regarded as optimal and should remain unchanged. This assumption aligns with the analysis proposed by \citet{Liu24}. The stage 1 augmented model is given by
    \begin{equation*}
    \begin{aligned}
    (Y_{i1},V_{i}^{(0:7)})
    \mid V_{i2}^{(0:7)},a_{i1},\bm{z}_{i1},\btheta_1 &\simind f_{\textsc{N}}(y+\max_{t}\{V_{i2}^{(t)}\};\btheta_1^\intercal [\bm{z}_{i1} \times \bm{a}_{i1}^*],\sigma_1^2)\prod_{t=0}^{7}f_{\textsc{N}}(v_t; \btheta_1^\intercal [\bm{z}_{i1} \times \bm{i}_t],\sigma_1^2).
    \end{aligned}
    \end{equation*}
 Finally, we assume a DSS prior for $(\btheta_1, \btheta_2)$. The defined model  includes a total of \( d_2 = (2T - 1)(k + 1) = 195 \) regressors in stage 2 and \( d_1 = T(k + 1) = 104 \) in stage 1. These figures are obtained by counting the \( k = 12 \) main effects, one intercept, and 182 and 91 treatment effects in stage 2 and stage 1, respectively. Moreover, having the same prognostic factors at both stages implies that $d = d_1$. We also note that, including interactions between covariates and the term $a_{i1}a_{i2}$ could be of interest, albeit at the cost of increased computational burden. As a preliminary step, we found that DSS provides a superior fit compared to ISS (see Table~\ref{bic and lpml} 
 in the Appendix). We thus present the analysis of the MIMIC-III data using the BAL model with normal regression functions \eqref{eq:normal_reg} and the DSS prior \eqref{eq:ss_prior1}, implemented by running Algorithm \ref{algo} 
 for 10,000 iterations, discarding the first 5,000 as burn-in. Results obtained by applying ISS and Q-learning with lasso and RF are summarized in Appendix \ref{app:clinical}. 
\begin{figure}
    \centering
    \includegraphics[width=0.9\linewidth,height=205pt]{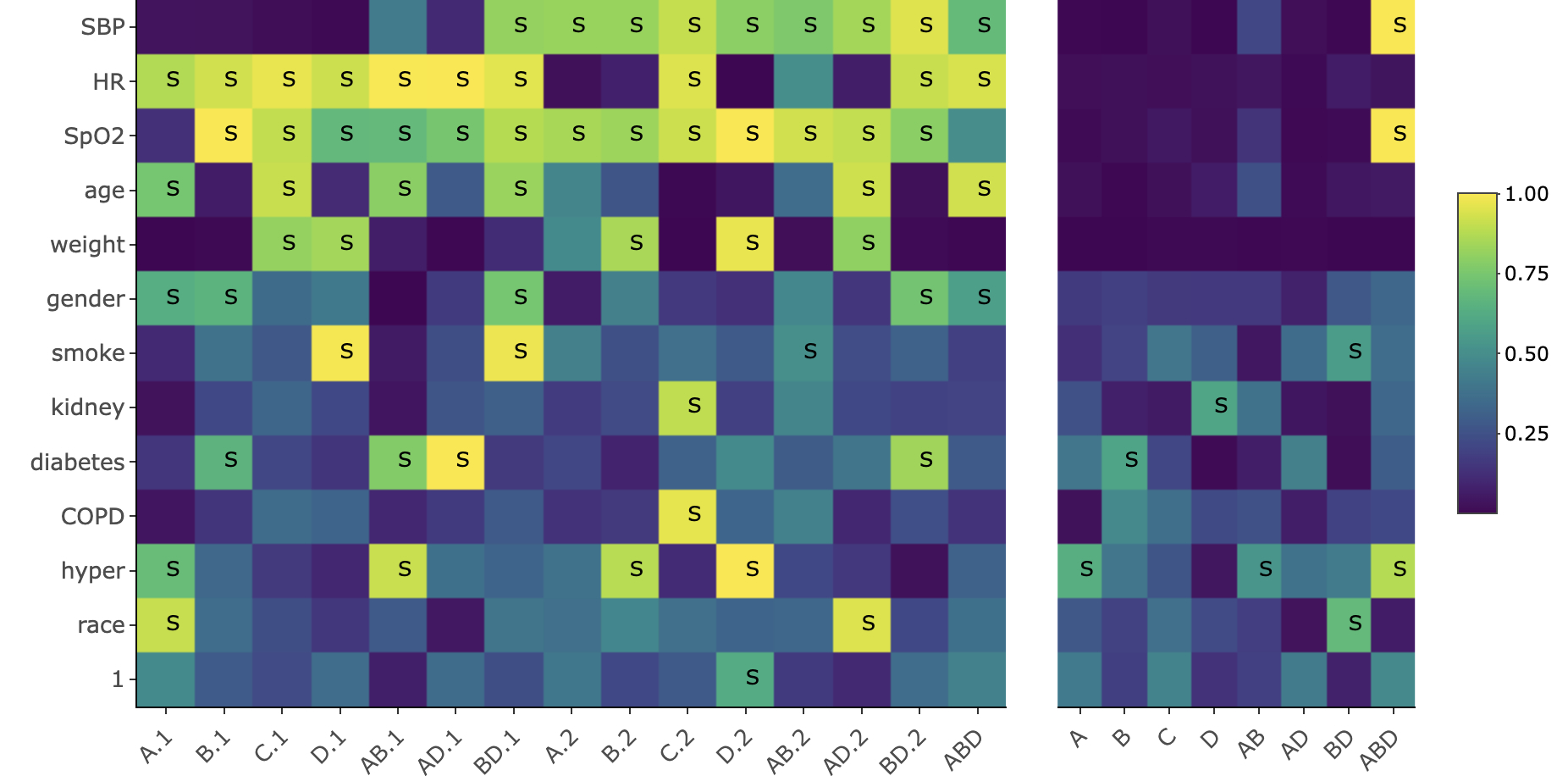}
\caption{MIMIC-III data. Posterior inclusion probabilities at stage 2 (left panel) and stage 1 (right panel). Treatment options are shown on the x-axis, while clinical factors are listed on the y-axis. A stands for ACEi, B for beta-blockers, C for CCB, and D for diuretics. AB stands for the combination of A and B, with other combinations defined similarly. The numeric suffixes 1 and 2 distinguish first and second-stage agents, e.g. A.1 and A.2. The symbol S, corresponding to posterior inclusion probabilities larger than 0.5, indicates statistical significance.}\label{figinteraction2_new}
\end{figure}
Studying interaction effects in this clinical study provides insight into how antihypertensive agents, or their combinations, interact with specific clinical factors to influence patients' SBP. Figure~\ref{figinteraction2_new} shows the posterior inclusion probabilities for the 12 clinical factors and their interactions with the eight treatment options.
In stage 1, a limited number of variables demonstrate statistical significance, with SBP and SpO2 showing notable effects, particularly in treatments involving a combination of ACEi, beta-blockers and diuretics. 
In stage 2, SBP, SpO2, smoking status, kidney, diabetes, chronic hypertension, and race remain significant across various drug combinations, highlighting their continued relevance. Furthermore, a broader range of clinical factors, including HR, age, weight, gender, and COPD, show higher posterior inclusion probabilities when interacting with treatments, suggesting their growing relevance in optimizing therapy as treatment progresses.
\begin{figure}
\centering
\includegraphics[clip,trim=0.6cm 1.5cm 1.5cm 2cm,width=0.31\textwidth]{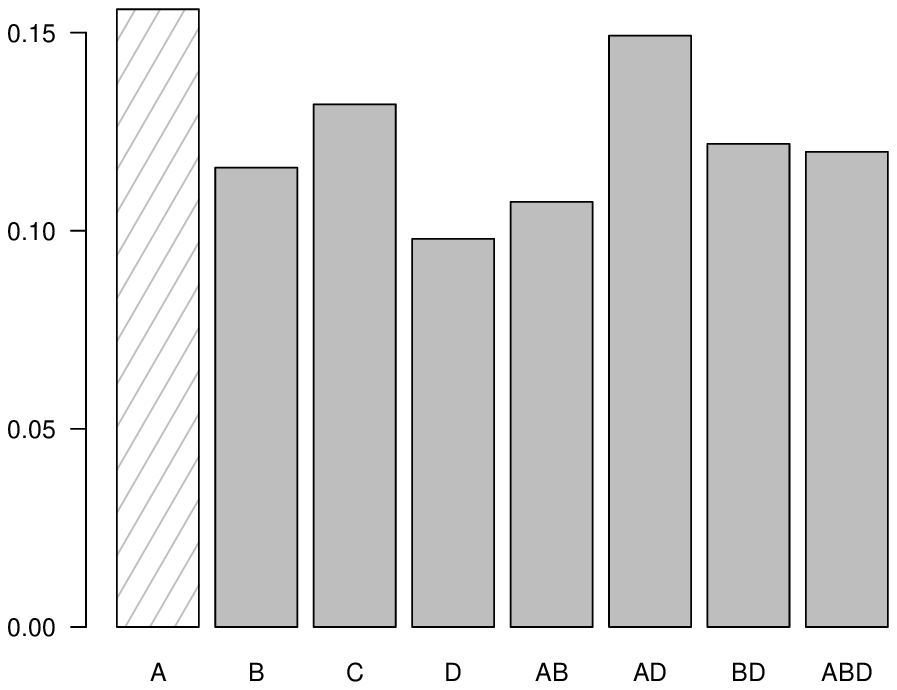}\hspace{0.5cm}\includegraphics[clip,trim=0.6cm 1.5cm 1.5cm 2cm,width=0.31\textwidth]{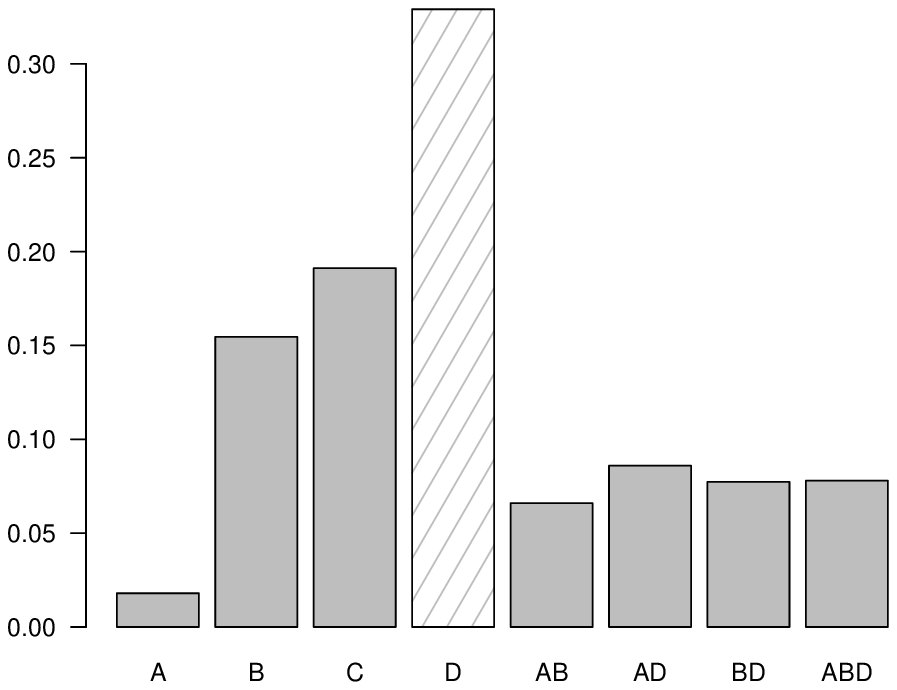}\hspace{0.5cm}\includegraphics[clip,trim=0.6cm 1.5cm 1.5cm 2cm,width=0.31\textwidth]{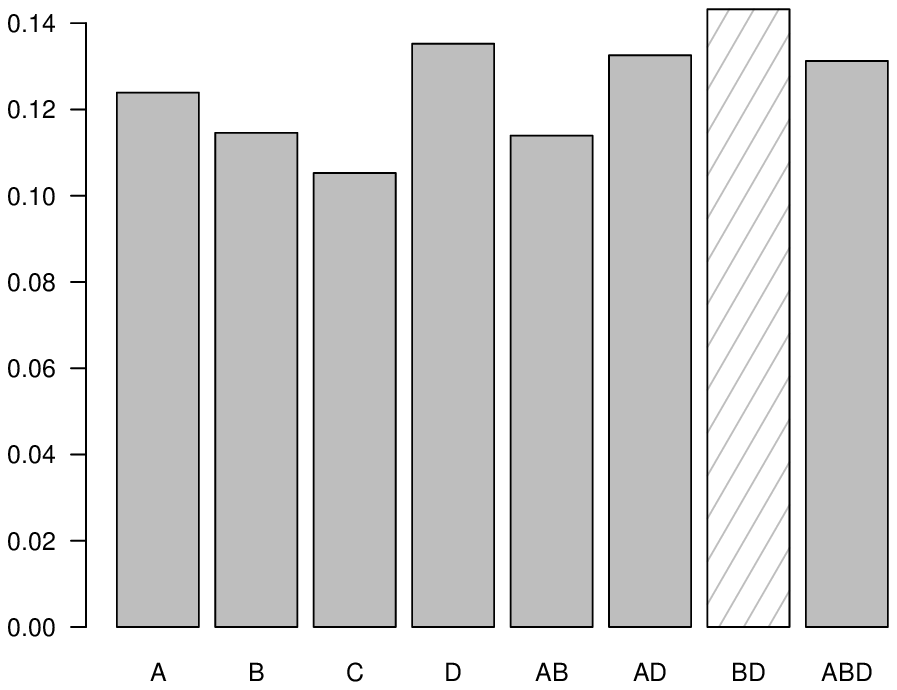}
\caption{MIMIC-III data. Posterior probability of each of the \(T=8\) treatment options being the optimal choice for Individual \#1 at stage 1 (left panel), for Individual \#2 at stage 2 (middle panel), and for Individual \#2 at stage 1 (right panel). A stands for ACEi, B for beta-blockers, C for CCB, and D for diuretics. AB stands for the combination of A and B, with other combinations defined similarly. In each barplot, the white bin indicates the largest posterior probability.}
\label{second patient combination_new_1}
\end{figure}
To further illustrate the utility of our method, we present additional results for two randomly selected patients, referred to as Individuals \#1 and \#2.
Figure~\ref{second patient combination_new_1} presents the stage-specific probabilities of each treatment being optimal, considering clinical profiles and treatment history. For Individual \#1, who received treatment only on day one, ACEi emerges as the most effective option. For Individual \#2, treated over two days, the optimal regimen is beta-blockers and diuretics at stage 1, followed by diuretics alone at stage 2.

\section{Conclusion}\label{section6}
We introduced BAL, a novel Bayesian framework for DTRs that integrates uncertainty quantification and variable selection, while favoring information sharing across stages. In a comprehensive simulation study, BAL outperformed the considered alternatives in identifying optimal treatment sequences and key variables, and its DSS variant proved able to share information across stages when convenient. We further demonstrated the practical value of BAL with the DSS prior, through an application to ICU data on patients with severe acute arterial hypertension. The promising results of this work open several avenues for extending our approach. One direction is to refine the DSS prior so to allow for information sharing not only across stages for the same variable, but also among different variables within and across stages. Another direction involves introducing more flexible nonparametric regression models, with the DSS prior adapted to the inner spike-and-slab Bayesian nonparametric framework \citep{Can17, Can23}. Finally, extending the approach to settings with more than two stages presents an interesting challenge, likely requiring new computational strategies to manage the increasing complexity of counterfactual treatment sequences.


\section*{Acknowledgements}
Matteo Borrotti and Bernardo Nipoti have been supported by Fondazione Cariplo, grant n° 2024-0140.


\section*{Data Availability}

The data supporting the findings in this paper are available from PhysioNet and were used under a licensing agreement. Due to access restrictions, they are available only with PhysioNet’s authorization at https://physionet.org/content/mimiciii/1.4/.


\allowdisplaybreaks

\clearpage
\begin{appendices}

\numberwithin{figure}{section}
\numberwithin{table}{section}
\numberwithin{algorithm}{section}
\numberwithin{proposition}{section}
\numberwithin{equation}{section}

\begin{center}
{\LARGE Appendix}
\end{center}

\noindent This appendix is organized as follows. Section \ref{sec:balAPP} provides more details on the BAL model. Specifically, in Section \ref{sec:appA}, we discuss the validity of the equivalences presented in \eqref{eq:alt_opt1} and provide proofs for specific cases. In Section \ref{sec:appASS} we discuss the stable unit treatment value assumption and the assumption of no unmeasured confounders. Section \ref{sec:corr_derivation} outlines the derivation of the expression in \eqref{eq:corr}, which pertains to the correlation between the indicator variables controlling the significance of the same regressor across the two stages. Section \ref{sec:appB} offers additional details on posterior computations. Sections \ref{app:sim} and \ref{app:clinical} present supplementary results on the simulation study and the analysis of the MIMIC-III data, respectively. 

\section{More details on the BAL model}\label{sec:balAPP}
\subsection{On the validity of the equivalences in \eqref{eq:alt_opt1}}\label{sec:appA}

We discuss and prove the validity of the equivalences in \eqref{eq:alt_opt1} for some cases of interest. As in Section~\ref{section2}, $Y_{i2}$ and $\bar{Y}_{i2}$ are assumed independent and distributed according to $f_2$, as specified in \eqref{eq:model1}. Similarly, $Y_{i}^\opt$ and $\bar{Y}_{i}^\opt$ are assumed independent and distributed according to $f_1$. We start by proving Proposition~\ref{proposition1}, which refers to the normal specification of $f_2$ and $f_1$, described in Section~\ref{sec:reg} and implemented in the paper. Propositions~\ref{proposition2} and~\ref{proposition3} showcase that the same approach used in this work can be used for other specifications of $f_2$ and $f_1$, such as logistic and Poisson regressions.

\begin{proof}[Proof of Proposition \ref{proposition1}]
Given the continuity of the payoffs $Y_{i2}$ and $\bar{Y_{i2}}$, proving \eqref{eq:alt_opt1} is equivalent to proving \eqref{eq:alt_opt3}. Combining the assumptions on $Y_{i2}$ and $\bar{Y}_{i2}$, we obtain that $Y_{i2}-\bar{Y}_{i2}\sim \textsc{N}(\mu_1-\mu_2,2\sigma_2^2)$, where $\mu_1=\btheta_2^\intercal [\bm{z}_{i2} \times (a_{i1},a_{i2})]$ and $\mu_2=\btheta_2^\intercal [\bm{z}_{i2} \times (a_{i1},1-a_{i2})]$. 
Thus, we can write
\begin{align*}
    \Pr(Y_{i2}-\bar{Y}_{i2}\geq 0)> 0.5 & \iff \Phi\left(\frac{\mu_1-\mu_2}{\sqrt{2\sigma_2^2}}\right)> 0.5\\
    &\iff \mu_1> \mu_2\\
    &\iff \mathbb{E}[Y_{i2}\mid a_{i1},a_{i2},\bm{z}_{i2},\btheta_2]> \mathbb{E}[Y_{i2}\mid a_{i1},1-a_{i2},\bm{z}_{i2},\btheta_2]\\
    &\stackrel{\footnotesize\text{(a)}}{\iff} a_{i2}^\opt=a_{i2},
\end{align*}
where (a) follows from \eqref{eq:decision_rule}.\\

\noindent Similarly for the equivalence involving $Y_{i}^\opt$ and $\bar{Y}_{i}^\opt$.
\end{proof}

\vspace{0.2cm}

\begin{proposition}\label{proposition2}
If $f_2$ and $f_1$ are logistic regressions, specified as 
\begin{equation}\label{eq:poi_reg}
\begin{aligned}
f_{2}(y; a_{i1},a_{i2},\bm{z}_{i2},\btheta_2)&=f_{\textsc{Ber}}(y; 1/(1+\exp\{-\btheta_2^\intercal [\bm{z}_{i2} \times (a_{i1},a_{i2})]\})),\\
     f_{1}(y; a_{i1},\bm{z}_{i1},\btheta_1)&=f_{\textsc{Ber}}(y; 1/(1+\exp\{-\btheta_1^\intercal [\bm{z}_{i1} \times a_{i1}]\})),
    \end{aligned}
\end{equation}
where $f_{\textsc{Ber}}(y;\pi)$ denotes the probability mass function of a Bernoulli random variable with mean $\pi$,
then the equivalences in \eqref{eq:alt_opt1} hold.
\end{proposition}
\begin{proof} We set $\pi_1=1/(1+\exp\{-\mu_1\})$ and $\pi_2=1/(1+\exp\{-\mu_2\})$, with $\mu_1=\btheta_2^\intercal [\bm{z}_{i2} \times (a_{i1},a_{i2})]$ and $\mu_2=\btheta_2^\intercal [\bm{z}_{i2} \times (a_{i1},1-a_{i2})]$. Combining the assumptions on $Y_{i2}$ and $\bar{Y}_{i2}$, we find that $Y_{i2}-\bar{Y}_{i2}$ can take values $\{-1,0,1\}$, with probabilities respectively equal to $\{(1-\pi_1)\pi_2,\pi_1\pi_2+(1-\pi_1)(1-\pi_2),\pi_1(1-\pi_2)\}$.
Thus, we have
\begin{align*}
    \frac{\Pr(Y_{i2}\geq\bar{Y}_{i2})}{\Pr(Y_{i2}\leq\bar{Y}_{i2})}>1 & \iff \frac{d_1(1-d_2)+d_1d_2+(1-d_1)(1-d_2)}{(1-d_1)d_2+d_1d_2+(1-d_1)(1-d_2)}>1\\
    &\iff \frac{1+\exp\{\mu_1\}+\exp\{\mu_1+\mu_2\}}{1+\exp\{\mu_2\}+\exp\{\mu_1+\mu_2\}}>1\\
    & \iff \mu_1>\mu_2\\
    &\iff \mathbb{E}[Y_{i2}\mid a_{i1},a_{i2},\bm{z}_{i2},\btheta_2]> \mathbb{E}[Y_{i2}\mid a_{i1},1-a_{i2},\bm{z}_{i2},\btheta_2]\\
    &\stackrel{\footnotesize\text{(a)}}{\iff} a_{i2}^\opt=a_{i2},
\end{align*}
where (a) follows from \eqref{eq:decision_rule}.\\

\noindent Similarly for the equivalence involving $Y_{i}^\opt$ and $\bar{Y}_{i}^\opt$.
\end{proof}

\begin{proposition}\label{proposition3}
If $f_2$ and $f_1$ are Poisson regressions, specified as 
\begin{equation}\label{eq:poi_reg}
\begin{aligned}
f_{2}(y; a_{i1},a_{i2},\bm{z}_{i2},\btheta_2)&=f_{\textsc{Poi}}(y; \btheta_2^\intercal [\bm{z}_{i2} \times (a_{i1},a_{i2})]),\\
     f_{1}(y; a_{i1},\bm{z}_{i1},\btheta_1)&=f_{\textsc{Poi}}(y; \btheta_1^\intercal [\bm{z}_{i1} \times a_{i1}]),
    \end{aligned}
\end{equation}
where $f_{\textsc{Poi}}(y; \lambda)$ denotes the probability mass function of a Poisson random variable with mean $\lambda$,
then the equivalences in \eqref{eq:alt_opt1} hold.
\end{proposition}
\begin{proof} Combining the assumptions on $Y_{i2}$ and $\bar{Y}_{i2}$, we obtain that $Y_{i2}-\bar{Y}_{i2}\sim \textsc{Skellam}(\mu_1,\mu_2)$,
where $\mu_1=\btheta_2^\intercal [\bm{z}_{i2} \times (a_{i1},a_{i2})]$ and $\mu_2=\btheta_2^\intercal [\bm{z}_{i2} \times (a_{i1},1-a_{i2})]$.
Thus, if we let $f_\textsc{Ske}(y;\mu_1,\mu_2)$ denote the probability mass function of a Skellam random variable with parameters $\mu_1$ and $\mu_2$, we can write
\begin{align*}
    \frac{\Pr(Y_{i2}\geq\bar{Y}_{i2})}{\Pr(Y_{i2}\leq\bar{Y}_{i2})}>1 & \iff \frac{\sum_{\ell=0}^\infty f_{\textsc{Ske}}(\ell;\mu_1,\mu_2)}{\sum_{\ell=-\infty}^0 f_{\textsc{Ske}}(\ell;\mu_2,\mu_1)}>1.\\
    & \iff \frac{\sum_{\ell=0}^\infty f_{\textsc{Ske}}(\ell;\mu_1,\mu_2)}{\sum_{\ell=0}^\infty f_{\textsc{Ske}}(-\ell;\mu_2,\mu_1)}>1\\
    & \stackrel{\footnotesize\text{(a)}}{\iff} \mu_1>\mu_2\\
    &\iff \mathbb{E}[Y_{i2}\mid a_{i1},a_{i2},\bm{z}_{i2},\btheta_2]> \mathbb{E}[Y_{i2}\mid a_{i1},1-a_{i2},\bm{z}_{i2},\btheta_2]\\
    &\stackrel{\footnotesize\text{(b)}}{\iff} a_{i2}^\opt=a_{i2},
\end{align*}
where (a) is implied by the fact that, for any $\ell>0$, 
\begin{equation*}
    \frac{f_{\textsc{Ske}}(\ell;\mu_1,\mu_2)}{f_{\textsc{Ske}}(-\ell;\mu_2,\mu_2)}=\left(\frac{\mu_1}{\mu_2}\right)^\ell,
\end{equation*}
and (b) follows from \eqref{eq:decision_rule}.\\

\noindent Similarly for the equivalence involving $Y_{i}^\opt$ and $\bar{Y}_{i}^\opt$.
\end{proof}

\subsection{SUTVA and NUC assumptions}\label{sec:appASS}
We briefly discuss the stable unit treatment value assumption (SUTVA) \citep{rubin1978bayesian} and the no unmeasured confounders (NUC) assumption \citep{rubin1980bias}, which are invoked in Section \ref{section2}. 
\begin{itemize}
    \item[(i)] The SUTVA asserts that the potential outcome for a given individual depends solely on their own treatment history and not on the treatment assignments of other individuals.     
    \item[(ii)] The NUC assumes that, at each stage, the treatment assignment is conditionally independent of future potential outcomes and covariates, given the observed history up to that point.  
\end{itemize}
Both the SUTVA and the NUC assumption are standard in the context of DTRs, as thoroughly and insightfully discussed in \citet{tsiatis2019dynamic}. These assumptions are fundamental in causal inference. Specifically and in the context of DTRs considered in our paper, both assumptions are required to identify the $j$th stage-specific average causal treatment effect, defined as $\mathbb{E}[Y_j(1)] - \mathbb{E}[Y_j(0)]$, when the treatment set is $\mathcal{A} = \{0,1\}$. Requiring that the distribution of each potential outcome $Y_j(1)$ and $Y_j(0)$ is independent of the treatment assignment mechanism is essential for reliably estimating the average causal treatment effect. While the primary focus of our work is not on causal inference, we implicitly rely on the SUTVA and NUC assumption when the models for the stage-specific pseudo-outcomes in Section \ref{section2}.

\subsection{Derivation of the expression in \eqref{eq:corr}}\label{sec:corr_derivation}
We show that $\rho(\delta_{l_11},\delta_{l_22})=
        1/(1+a+b)$
if $l_1=l_2\in\{1,\ldots,d\}$, and $0$ otherwise. $\delta_{l_11}$ and $\delta_{l_12}$ are Bernoulli random variables with parameters $w_{l_11}$ and $w_{l_22}$ respectively, and $w_{l_11} = w_{l_22}=w \sim \textsc{Beta}(a, b)$ for $l_1=l_2\in\{1,\ldots,d\}$. Since $\delta_{l_1 1}$ and $\delta_{l_2 2}$ are conditionally independent given $w$, we can express their joint conditional distribution as:
\[ \Pr(\delta_{l_1 1} = \alpha, \delta_{l_2 2} = \beta \mid w) = w^{\alpha+\beta} (1 - w)^{2 - (\alpha+\beta)}, \]
where $\alpha, \beta \in \{0, 1\}$. To find the unconditional joint probability, we marginalize with respect to $w$ and obtain
\[ \Pr(\delta_{l_1 1} = \alpha, \delta_{l_2 2} = \beta) = \int_0^1 \Pr(\delta_{l_1 1} = \alpha, \delta_{l_2 2} = \beta \mid w) f(w)  \dd w, \]
where $f(w)$ is the probability density function of the $\text{Beta}(a,b)$ distribution, that is
\[ f(w) = \frac{w^{a - 1} (1 - w)^{b - 1}}{\text{B}(a, b)},\]
and $\text{B}(a, b)$ is the Beta function. This leads to:
\begin{align*}\Pr(\delta_{l_1 1} = \alpha, \delta_{l_2 2} = \beta) &=\frac{\text{B}(\alpha+\beta+a, 2 - (\alpha+\beta) + b)}{\text{B}(a,b)}
\end{align*}
Marginally, we have 
\[ \mathbb{E}[\delta_{l_1,1}]=\Pr(\delta_{l_1 1} = 1) = \frac{\text{B}(a + 1, b)}{\text{B}(a,b)} = \frac{a}{a + b}, \]
and similarly for $\mathbb{E}[\delta_{l_2,2}]$. From this we get 
\[ \mathbb{E}[\delta_{l_1 1} \delta_{l_2 2}] = \Pr(\delta_{l_1 1} = 1, \delta_{l_2 2} = 1) = \frac{\text{B}(a + 2, b)}{\text{B}(a,b)} = \frac{a (a + 1)}{(a+b) (a+b + 1)}. \]
We can thus write the covariance between  $\delta_{l_1 1}$ and $\delta_{l_2 2}$ as
\begin{align*}
\text{Cov}(\delta_{l_1 1}, \delta_{l_2 2}) =\mathbb{E}[\delta_{l_1 1} \delta_{l_2 2}] - \mathbb{E}[\delta_{l_1 1}] \mathbb{E}[\delta_{l_2 2}]=\frac{a b}{(a + b)^2 (a + b + 1)}.
\end{align*}
The variance of $\delta_{l_1 1}$ is given by
\[ \text{Var}(\delta_{l_1 1}) = \mathbb{E}[\delta_{l_1 1}^2] - (\mathbb{E}[\delta_{l_1 1}])^2 = \frac{a}{a + b} \left(1 - \frac{a}{a + b}\right) = \frac{a b}{(a + b)^2}, \]
which coincides with the variance for $\delta_{l_2 2}$. As a result, we can write the correlation as
\[ \rho (\delta_{l_1 1}, \delta_{l_2 2}) = \frac{\frac{a b}{(a + b)^2 (a + b + 1)}}{\sqrt{\frac{a b}{(a + b)^2} \cdot \frac{a b}{(a + b)^2}}} = \frac{1}{1+a + b}. \]

\section{Posterior Inference}\label{sec:appB}
\subsection{Full conditional distributions}\label{sec:FC}
We next describe the full conditional distributions involved in the updates of Algorithm \ref{algo}, presented in Section \ref{sec:algo}.\\

\noindent\emph{Full conditional of $\btheta_{1}$:}
\begin{equation*}
    \btheta_{1}\mid \ldots\sim \textsc{N}_{d_1}\left(\bm{\mu}_{1}, \bm{\Sigma}_{1}\right),
\end{equation*}
where $$\bm{\Sigma}_{1}^{-1}=\frac{\bm{X}_1^{T}  \bm{X}_1}{\sigma_{1}^{2}}  +D_{1}^{-1}$$ and
$$
\bm{\mu}_{1}=\frac{\bm{\Sigma}_{1}  \bm{X}_1^{T}  \bY^{\text{\opt}}}{\sigma_{1}^2}.
$$
$D_{1}$ is a diagonal matrix with entries $r^{\left(1-\delta_{l_11}\right)}\psi_{l_11}, \ l_1=1,\ldots,d_1$, and $d_1$ is the dimension of $\btheta_{1}$,  
$$\bm{X}_{1}=\left(\begin{array}{c}
\bm{z}_{11} \times a_{11} \\
\ldots\\
\bm{z}_{n1} \times a_{n1}
\end{array}\right), \quad \bY^{\opt}=\left(\begin{array}{c}
Y_{11}+\max(V_{12}^{(0)},V_{12}^{(1)}) \\
\ldots\\
Y_{n1}+\max(V_{n2}^{(0)},V_{n2}^{(1)})
\end{array}\right).$$

\vspace{0.3cm}

\noindent\emph{Full conditional of $\btheta_{2}$:}
\begin{equation*}
    \btheta_{2}\mid \ldots \sim \textsc{N}_{d_2}\left(\bm{\mu}_{2}, \bm{\Sigma}_{2}\right),
\end{equation*}
where $$\bm{\Sigma}_{2}^{-1}=\frac{\bm{X}^{*T}  \bm{X}^{*}}{\sigma_{2}^{2}}+D_{2}^{-1}$$ and
$$
\bm{\mu}_{2}=\frac{\bm{\Sigma}_{2} \bm{X}^{*T} \bY^{*}}{\sigma_{2}^2}.
$$
$D_{2}$ is a diagonal matrix with entries $r^{\left(1-\delta_{l_22}\right)}\psi_{l_22}, \ l_2=1,\ldots,d_2$, and $d_2$ is the dimension of $\btheta_{2}$. Moreover,
$$\bm{X}^{*}=\left(\begin{array}{c}
\bm{z}_{12} \times (a_{11},a_{12}) \\
\ldots\\
\bm{z}_{n2} \times (a_{n1},a_{n2})\\
\bm{z}_{12} \times (a_{11},1-a_{12})\\
\ldots\\
\bm{z}_{n2} \times (a_{n1},1-a_{n2})
\end{array}\right), \quad \bY^{*}=\left(\begin{array}{c}
Y_{12} \\
\ldots\\
Y_{n2} \\
V_{12}^{(1-a_{12})} \\
\ldots\\
V_{n2}^{(1-a_{n2})}\\
\end{array}\right).$$   

\noindent\emph{Full conditional of $\sigma_{1}^2$:}
\begin{equation*}
    \sigma_{1}^{2}\mid \ldots \sim \textsc{inv-Gamma} \left(\frac{n+1}{2}, \beta+\frac{1}{2}\right),
\end{equation*}
where $\beta=\frac{1}{2} \sum_{i=1}^{n}\left(Y_{i1}-\left(\btheta_1^T [\bm{z}_{i1} \times a_{i1}]-\max\left(V_{i2}^{(0)},V_{i2}^{(1)}\right)\right)\right)^2$ is the scale parameter of the inverse gamma distribution.\\

\noindent\emph{Full conditional of $\sigma_{2}^2$:}
\begin{equation*}
    \sigma_{2}^{2}\mid \ldots \sim \textsc{inv-Gamma} \left(n+\frac{1}{2},\beta+\frac{1}{2}\right),
\end{equation*}
where $$\beta=\frac{1}{2} \sum_{i=1}^{n}\left(\left(Y_{i2}-\btheta_2^T [\bm{z}_{i2} \times (a_{i1}, a_{i2})]\right)^2+\left(V_{i2}^{(1-a_{i2})}-\btheta_2^T [\bm{z}_{i2} \times (a_{i1}, 1-a_{i2})]\right)^2\right)$$ is the scale parameter of the inverse gamma distribution.\\

\noindent\emph{Full conditionals for $\delta_{l_11}$ and $\delta_{l_22}$:} for any $l_1=1,\ldots,d_1$
\begin{equation*}
    \delta_{l_11}\mid \ldots \sim \textsc{Bern}\left(\frac{w_{l_11} f_{\textsc{N}}(\theta_{l_11};0,\psi_{l_11})}{w_{l_11} f_{\textsc{N}}(\theta_{l_11};0,\psi_{l_11})+(1-w_{l_11}) f_{\textsc{N}}(\theta_{l_11};0,r\psi_{l_11})}\right).
\end{equation*}
Similarly, for any $l_2=1,\ldots,d_2$,
\begin{equation*}
    \delta_{l_22}\mid \ldots \sim \textsc{Bern}\left(\frac{w_{l_22} f_{\textsc{N}}(\theta_{l_22};0,\psi_{l_22})}{w_{l_22} f_{\textsc{N}}(\theta_{l_22};0,\psi_{l_22})+(1-w_{l_22}) f_{\textsc{N}}(\theta_{l_22};0,r\psi_{l_22})}\right).
\end{equation*}

\noindent\emph{Full conditionals for $\psi_{l_11}$ and $\psi_{l_22}$:} for any $l_1=1,\ldots,d_1$
\begin{equation*}
    \psi_{l_11}\mid \ldots \sim \textsc{inv-Gamma}\left(v+\frac{1}{2}, Q+\frac{\theta^{2}_{l_11}}{2r^{(1-\delta_{l_11})}}\right).
\end{equation*}
Similarly, for any $l_2=1,\ldots,d_2$,
\begin{equation*}
    \psi_{l_22}\mid \ldots \sim \textsc{inv-Gamma}\left(v+\frac{1}{2}, Q+\frac{\theta^{2}_{l_22}}{2r^{(1-\delta_{l_22})}}\right).
\end{equation*}

\noindent\emph{Full conditional for $w_{l1}=w_{l2}$:} for every $l=1,\ldots,p$
    \begin{equation*}
        w_{l1}\mid \ldots \sim \textsc{Beta}\left(a+\delta_{l1}+\delta_{l2},b+2-\delta_{l1}-\delta_{l2}\right).
    \end{equation*}
\noindent\emph{Full conditionals for $w_{l_11}$ and $w_{l_22}$:} for any $l_1=p+1,\ldots,d_1$:   
     \begin{equation*}
        w_{l_11}\mid \ldots \sim \textsc{Beta}\left(a+\delta_{l_11},b+1-\delta_{l_11}\right).
    \end{equation*}
Similarly, for any $l_2=p+1,\ldots,d_2$  
    \begin{equation*}
        w_{l_22}\mid \ldots \sim \textsc{Beta}\left(a+\delta_{l_22},b+1-\delta_{l_22}\right).
    \end{equation*}

\noindent\emph{Full conditional for the pseudo-outcomes $\bm{V}_2$ and $\bm{V}$:} for $i=1,\ldots,n$, we have
 \begin{align*}
    V_{i2}^{(0)}\mid \ldots &\sim \textsc{N}(\btheta_2^T[\bz_{i2}\times(a_{i1},0)],\sigma_2),\\
V_{i2}^{(1)}\mid \ldots &\sim \textsc{N}(\btheta_2^T[\bz_{i2}\times(a_{i1},1)],\sigma_2),\\
V_{i}^{(0)}\mid \ldots &\sim \textsc{N}(\btheta_1^T[\bz_{i1}\times 0],\sigma_1),\\
V_{i}^{(1)}\mid \ldots &\sim \textsc{N}(\btheta_1^T[\bz_{i1}\times 1],\sigma_1).
\end{align*}

\noindent\emph{Full conditionals for the hyper-parameters $a$ and $b$:}
\begin{equation*}
a \mid \ldots \sim \left(\frac{\Gamma(a+b)}{\Gamma(a)}\right)^{d_1+d_2-d}
\left(\prod_{l_1=1}^{d_1} w_{l_11} \prod_{l_2=d+1}^{d_2} w_{l_22}\right)^{(a-1)}p(a),
\end{equation*} 
\begin{equation*}
b \mid \ldots \sim \left(\frac{\Gamma(a+b)}{\Gamma(b)}\right)^{d_1+d_2-p}
\left(\prod_{l_1=1}^{d_1} (1-w_{l_11}) \prod_{l_2=d+1}^{d_2} (1-w_{l_22})\right)^{(b-1)}p(b),
\end{equation*}
where $p(a)$ and $p(b)$ represents the prior distributions of $a$ and $b$, in this work assumed as inverse gamma with both shape and scale parameters equal to 1. To sample from the full conditional distributions of $a$ and $b$, we utilize a random-walk Metropolis-Hastings algorithm \citep{chib1995understanding}, with uniform proposal centered at the current value of the hyper-parameter, with variances chosen to achieve optimal acceptance rates \citep{Gel97}.

\subsection{Algorithm}\label{sec:algo}
The joint distribution of observations $\bY_1=(Y_{11},\ldots,Y_{n1})$ and $\bY_2=(Y_{12},\ldots,Y_{n2})$, pseudo-outcomes $\bm{V}_2=(V_{12}^{(0)},V_{12}^{(1)},\ldots,V_{n2}^{(0)},V_{n2}^{(1)})$ and $\bm{V}=(V_{1}^{(0)},V_{1}^{(1)},\ldots,V_{n}^{(0)},V_{n}^{(1)})$, and parameters $\varphi=(\btheta_1,\btheta_2,\allowbreak \bdelta_1,\bdelta_2,\bpsi_1,\bpsi_2,\bw_1,\bw_2,\sigma_1^2,\sigma_2^2)$ is described by combining the BAL model in \eqref{eq:aug_model}, completed with normal regression functions as in \eqref{eq:normal_reg} and DSS prior as in \eqref{eq:ss_prior1}, with the specification of a normal prior with independent components for $(\sigma_1^2,\sigma_2^2)$. From this, we devise a Gibbs sampling scheme for the posterior distribution $p(\varphi,\bm{V}_2,\bm{V}\mid \by_1,\by_2)$, from which it is straightforward to estimate the posterior distribution $p(\varphi\mid \by_1,\by_2)$ and its functionals. The steps of the Gibbs sampler are summarized in Algorithm~\ref{algo} and involve a sequential update of the parameters characterizing the two stage-specific regressions and the pseudo-outcomes. Thanks to the conjugate specification of the priors, the full conditional distributions involved in the algorithm are simple to sample from, as detailed in Section \ref{sec:FC}. 
Again, the extension to the case with $T$ treatment options at each stage is straightforward.  

\begin{algorithm}
\caption{Gibbs sampler for BAL with DSS}\label{algo}
\begin{algorithmic}[1]
\State Initialize parameters $\varphi^{(0)}$ \Comment{Superscript $(\ell)$ to denote iteration $\ell$}
\State Initialize pseudo-outcomes $(\bm{V}_{2}^{(0)},\bm{V}^{(0)})$
\For{$\ell = 1$ to $L$} \Comment{$L$: number of iterations}
\State Update parameters $\varphi^{(\ell)}$;\vspace{0.05cm}
\State Update pseudo-outcomes $(\bm{V}_{2}^{(\ell)},\bm{V}^{(\ell)})$;
\State Compute $(A_{i2}^{\opt (\ell)},A_{i1}^{\opt (\ell)})$, for $i=1,\ldots,n$.
\EndFor 
\State Set $\hat{a}_{ij}^\opt$ as mode of $\{A_{ij}^{\opt(L_0+1)},\ldots,A_{ij}^{\opt(L)}\}$, for $j=1,2$  \Comment{$L_0$: burn-in iterations}
\end{algorithmic}
\end{algorithm}

The hyper-parameters $a$ and $b$ control both the prior distribution of the probability of inclusion of a regressor, and, as apparent in \eqref{eq:corr}, the correlation across the significance of common regressors across stages. We assign them a prior and discuss, in the Supplementary Material, a Metropolis--Hastings step to deal with their update in Algorithm \ref{algo}.

\section{Additional details on the Simulation Study}\label{app:sim}
\subsection{Definition of the measures of performance}\label{sec:measures}
We describe the measures used to evaluate model accuracy in identifying significant regressors across stages, as well as the impact of variable selection on correctly determining the optimal DTR.\\
As far as variable selection is concerned, model performance is assessed by resorting to three evaluation indices: proportion of false negatives (FN), that is regressors that are incorrectly deemed as not significant; proportion of false positives (FP), that is regressors incorrectly identified as significant; and F$_1$ score, combining precision and recall into a single metric. At the same time, the capability of a model to correctly identify the optimal treatment is assessed by means of the treatment error rate (ER) and the payoff mean absolute error (MAE). 
While the ER only takes into account the actual optimality of the treatment identified as optimal, the MAE also accounts for the magnitude of the variation in terms of potential payoff, when a suboptimal treatment is mistakenly identified as optimal. When specific to the $j$th stage, the treatment ER and the MAE are defined as
\begin{equation*}
    \text{ER}_j=\frac{1}{n}\sum_{i=1}^{n}\mathbbm{1}{\{\hat{a}_{ij}^\opt\ne a_{ij}^\opt\}},\quad 
    \text{MAE}_j=\frac{1}{n}\sum_{i=1}^{n}|\mathbb{E}[Y_{ij}^\opt]-\mathbb{E}[Y_{ij}^{\widehat{\opt}}]|,
\end{equation*}
where $\mathbb{E}[Y_{ij}^\opt]$ and $\mathbb{E}[Y_{ij}^{\widehat{\opt}}]$ are the expected payoffs at the $j$th stage, respectively under the assumptions that the assigned treatments are, respectively, optimal and identified as optimal. The analogue overall measures of model performance are defined as 
\begin{equation*}
    \text{ER}=\frac{1}{n}\sum_{i=1}^{n}\mathbbm{1}{\{(\hat{a}_{i1}^\opt,\hat{a}_{i2}^\opt)\ne (a_{i1}^\opt,a_{i2}^\opt)\}},\quad 
    \text{MAE}=\frac{1}{n}\sum_{i=1}^{n}|\mathbb{E}[Y_{i}^\opt]-\mathbb{E}[Y_{i}^{\widehat{\opt}}]|.
\end{equation*}

\subsection{Data generating processes}\label{sec:dgp}
Throughout the simulation study of Section \ref{section4}, two data generating processes are considered, both defined as minor modifications of the class of generative models developed by \cite{chakraborty2010inference}, and later used in \cite{chakraborty2013inference} and \cite{laber2014dynamic}. We also introduce a third data generating process that differs from the previous ones in that the payoffs are not generated according to a linear model. 
\subsubsection{First experiment}\label{sec:dgp1}
A data generating process for synthetic data consisting of stage-specific covariates $\bm{z}_{ij}$, treatments $a_{ij}$, and stage-specific payoffs $y_{ij}$, is defined by means of the following chain of conditional distributions:
\begin{equation}\label{eq:dgp1}
\begin{aligned}
  &A_{ij}\simiid \textsc{Bern}(\{0,1\};0.5)&\text{(treatments)} \\
   &Z_{i1l}\simiid \textsc{Bern}(\{-1,1\};0.5), \text{ for }l=1,\ldots,k&\text{(covariates)} \\
   &Z_{i2l}\mid z_{i1l}\simind \textsc{Bern}(\{-1,1\};1/(1+\exp\{-z_{i1l}\}), \text{ if }l=1,\ldots,\floor{k/2} & \\ 
 &Z_{i2l}=z_{i1l}, \text{ if }l=\floor{k/2}+1,\ldots,k&\\
 &Y_{i1}\mid a_{i1},\bm{z}_{i1},\bm{z}_{i2},\btheta_{1}^*,\btheta_{2}^{*} \simind  \textsc{N}(\btheta_{1}^{*\intercal}[\bm{z}_{i1}\times a_{i1}]-\mathbb{E}[Y_{i2}^\opt],1)&\text{(payoffs)}\\
&Y_{i2}\mid a_{i1},a_{i2},\bm{z}_{i2},\btheta_{2}^{*} \simind \textsc{N}(\btheta_{2}^{*\intercal}[\bm{z}_{i2}\times(a_{i1},a_{i2})],1), &
\end{aligned}
\end{equation}
where $\mathbb{E}[Y_{i2}^\opt]=\max(\btheta_{2}^{*\intercal}[\bm{z}_{i2}\times(a_{i1},0)],\btheta_{2}^{*\intercal}[\bm{z}_{i2}\times(a_{i1},1)])$ and the notation $X\sim\textsc{Bern}(\{a,b\};\break \pi)$ is used to indicate $\Pr(X=b)=1-\Pr(X=a)=\pi$. According to \eqref{eq:dgp1}, covariates $\bm{z}_{ij}$ are generated to mimic a realistic scenario where about half remain constant across stages, while the values taken by the other half evolve from stage 1 to stage 2. To generate data where the significance of predictors is correlated across stages, we define the distribution of the regression coefficients \(\btheta_{1}^{*}\) and \(\btheta_{2}^{*}\), with dimensions \(d_1=2(k+1)\) and \(d_2=3(k+1)\), respectively, 
mimicking the structure of the DSS prior in \eqref{eq:ss_prior1}. Namely, 
\begin{equation}\label{eq:dgp_ss}
      \begin{aligned}
    \omega_{l1}^{*}=\omega_{l2}^{*}&\simiid \textsc{Beta}(a^{*},b^{*}), \text{ for }l=1,\ldots,d_1;\\
        \omega_{l2}^{*}&\simiid \textsc{Beta}(a^{*},b^{*}), \text{ for }l=d_1+1,\ldots,d_2;\\
        (\delta_{l_11}^{*},\delta_{l_22}^{*})\mid\omega_{l_11}^{*},\omega_{l_22}^{*}&\simind \textsc{Bern}(\omega_{l_11}^{*})\times\textsc{Bern}(\omega_{l_22}^{*}),\text{ for }l_j=1,\ldots,d_j;\\ 
      (m_{l_11}^{*},m_{l_22}^{*}) &\simind \textsc{Bern}\{-3,3\}\times \textsc{Bern}\{-3,3\},\text{ for }l_j=1,\ldots,d_j;\\
        (\tilde{\theta}_{l_11}^{*},\tilde{\theta}_{l_22}^{*})\mid m_{l_11}^*,m_{l_22}^* &\simind \textsc{N}(m_{l_11},1)\times \textsc{N}(m_{l_22},1),\text{ for }l_j=1,\ldots,d_j;\\
        \theta_{l_jj}^*&=\tilde{\theta}_{l_jj}^*\delta_{l_jj}^*, \text{ for }l_j=1,\ldots,d_j.
    \end{aligned}
\end{equation}

The three levels of correlation $\rho^*=\rho(\delta_{l1}^*,\delta_{l2}^*)$ between the indicator variables $\delta_{i1}^*$ and $\delta_{i2}^*$, when $l\in\{1,\ldots,d_1\}$, in the data generating process, that is weak ($\rho^*=0.3$), medium ($\rho^*=0.6$) and strong ($\rho^*=0.9$), are obtained through the specification of the parameters $(a^*,b^*)$ in \eqref{eq:dgp_ss}: by applying \eqref{eq:corr}, we have that $(a^*, b^*) \in \{ \left( 7/10, 49/30 \right), \left( 1/5, 7/15 \right), \left( 1/30, 7/90 \right) \}$ implies $\rho^* = \{0.3, 0.6, 0.9\}$, respectively. Each pair $(a^*, b^*)$ corresponds to an expected proportion of stage-specific significant predictors equal to $\mathbb{E}[\omega_{lj}^*]=0.3$, for $l=1,\ldots,d_j$ and $j=1,2$.
\subsubsection{Second experiment}\label{sec:dgp2}
Similarly to \eqref{eq:dgp1}, a data generating process for synthetic data consisting of stage-specific covariates $\bm{z}_{ij}$, treatments $a_{ij}$, and stage-specific payoffs $y_{ij}$, is defined by means of the following chain of conditional distributions:
\begin{align*}
  &A_{ij}\simiid \textsc{Unif}(\mathcal{A}_j)&\text{(treatments)} \\
   &A^*_{ijt}=\mathbb{I}_{\{t\}}(A_{ij}), \text{ for }t=1,\ldots,T-1&\text{(dummy variables)} \\
   &Z_{i1l}\simiid \textsc{Bern}(\{-1,1\};0.5), \text{ for }l=1,\ldots,k&\text{(covariates)}\\
   &Z_{i2l}\mid z_{i1l}\simind \textsc{Bern}(\{-1,1\};1/(1+\exp\{-z_{i1l}\}), \text{ if }l=1,\ldots,\floor{k/2}&\\
 &Z_{i2l}=z_{i1l}, \text{ if }l=\floor{k/2}+1,\ldots,k&\\
 &Y_{i1}\mid \bm{a}^*_{i1},\bm{z}_{i1},\bm{z}_{i2},\btheta_{1}^*,\btheta_{2}^{*} \simind  \textsc{N}(\btheta_{1}^{*\intercal}[\bm{z}_{i1}\times \bm{a}^*_{i1}]-\mathbb{E}[Y_{i2}^\opt],1)&\text{(payoffs)}\\
&Y_{i2}\mid \bm{a}^*_{i1},\bm{a}^*_{i2},\bm{z}_{i2},\btheta_{2}^{*} \simind \textsc{N}(\btheta_{2}^{*\intercal}[\bm{z}_{i2}\times(\bm{a}^*_{i1},\bm{a}^*_{i2})],1), &
    \end{align*}
where $\mathbb{E}[Y_{i2}^\opt]=\max_{\bm{a}^*\in \mathcal{D}_{T-1}}(\btheta_{2}^{*\intercal}[\bm{z}_{i2}\times(\bm{a}^*_{i1},\bm{a}^*)])$; $\mathbb{I}_{\{t\}}(A_{ij})$ denotes the indicator function, equal to $1$ if $A_{ij}=t$ and 0 otherwise; $\textsc{Unif}(\mathcal{A}_j)$ denotes the discrete uniform distribution on $\mathcal{A}_j$, and $\bm{a}^*_{ij}=(a_{ij1}^*,\ldots,a_{ij(T-1)}^*)\in \mathcal{D}_{T-1}$ indicates the vector of dummy variables encoding the treatment assigned to the $i$th individual at the $j$th stage. Moreover, $\mathcal{D}_{T-1}$ denotes the set of $(T-1)$-dimensional vectors consisting of at most one element equal to one and the remaining equal to zero. The vectors $\btheta_1^*$ and $\btheta_2^*$ are simulated by following \eqref{eq:dgp_ss} with $\rho^*=0.6$ as a result of $(a^*,b^*)=(1/5,7/15)$, and with $d_1=T(k+1)$ and $d_2=(2T-1)(k+1)$. The expected proportion of stage-specific significant predictors is 0.3.

\subsubsection{Third experiment}\label{sec:dgp3}
We consider an additional simulation scenario in which the payoffs \(y_{i1}\) and \(y_{i2}\) are not generated from normal linear regressions. Instead, they arise from a more complex partially linear data generating process in which the effect of half of the predictors is linear, while the effect of the remaining half is nonlinear. Specifically, synthetic data consisting of stage-specific covariates \(\bm{z}_{ij}\), treatments \(a_{ij}\), and stage-specific payoffs \(y_{ij}\) are generated according to the following sequence of conditional distributions:
\begin{equation}\label{eq:dgp3}
\begin{aligned}
  &A_{ij}\simiid \textsc{Bern}(\{0,1\};0.5)&\text{(treatments)} \\
  &Z_{i1l}\simiid \textsc{Bern}(\{0,1\};0.5), \text{ for }l=1,\ldots,\floor{k/2}&\text{(discrete covariates)}\\
   &Z_{i2l}\mid z_{i1l}\simind \textsc{Bern}(\{0,1\};1/(1+\exp\{-z_{i1l}\}), \text{ if }l=1,\ldots,\floor{k/4}&\\
 &Z_{i2l}=z_{i1l}, \text{ if }l=\floor{k/4}+1,\ldots,\floor{k/2}&\\
   &Z_{i1l}\simiid \textsc{N}(0,1), \text{ for }l=\floor{k/2}+1,\ldots,k&\text{(continuous covariates)} \\
   &Z_{i2l}\mid z_{i1l}\simind \textsc{N}(z_{i1l},1), \text{ if }l=\floor{k/2}+1,\ldots,\floor{3k/4} & \\ 
 &Z_{i2l}=z_{i1l}, \text{ if }l=\floor{3k/4}+1,\ldots,k&\\
 &Y_{i1}\mid a_{i1},\bm{z}_{i1},\bm{z}_{i2},\btheta_{1}^*,\btheta_{2}^{*} \simind  \textsc{N}(\btheta_{1}^{*\intercal}h\left([\bm{z}_{i1}\times a_{i1}]\right)-\mathbb{E}[Y_{i2}^\opt],1)&\text{(payoffs)}\\
&Y_{i2}\mid a_{i1},a_{i2},\bm{z}_{i2},\btheta_{2}^{*} \simind \textsc{N}(\btheta_{2}^{*\intercal}h\left([\bm{z}_{i2}\times(a_{i1},a_{i2})]\right),1), &
\end{aligned}
\end{equation}
where 
\begin{equation}\label{eq:h}
h(\tilde{z}_1,\ldots,\tilde{z}_d)=\left(\mathbb{I}_{[0.5,\infty)}(\tilde{z}_1),\ldots,\mathbb{I}_{[0.5,\infty)}(\tilde{z}_d)\right).
\end{equation}
According to \eqref{eq:dgp3} and similarly to \eqref{eq:dgp1}, the vector $\bm{z}_{ij}$ is generated to mimic a realistic setting in which roughly half of the covariates remain constant across stages, while the others evolve from stage 1 to stage 2. Moreover, a tree-like mechanism such as that induced by the function \(h\) in \eqref{eq:h} yields a setting in which the conditional expected payoffs are piecewise constant, with discontinuities occurring at the indicator thresholds of the covariates and at the interaction terms between treatment and covariates. Consequently, the treatment effect can vary across these regions whenever interaction terms are present, and the optimal treatment rule consists of decision boundaries at the indicator thresholds. This applies to both discrete and continuous covariates. We note, however, that the transformation induced by \(h\) is non-trivial only for continuous covariates and the corresponding interaction terms, since for discrete covariates and their interactions the mapping \(\tilde{z}\mapsto \mathbb{I}_{[0.5,\infty)}(\tilde{z})\) coincides with the identity. Also in this case, the vectors $\btheta_1^*$ and $\btheta_2^*$ are simulated by following \eqref{eq:dgp_ss} with $\rho^*\in\{0.3,0.6,0.9\}$ as a result of $(a^*,b^*)\in\{(7/10,49/30),(1/5,7/15),(1/30,7/90)\}$, and with $d_1=T(k+1)$ and $d_2=(2T-1)(k+1)$. 

\subsection{Implementation of alternative methods}
We provide details on the implementation of the alternative methods used for comparison in the data analyses of Sections \ref{section4} and \ref{section5}.
\label{sec:alt_meth}\begin{itemize}
    \item[i)] Q-learning with lasso is implemented by selecting the optimal penalty parameter through 4-fold cross-validation using the \texttt{glmnet} R package \citep{Fri21}.
\item[ii)] Q-learning with RF is implemented using the \texttt{R} package \texttt{randomForest} \citep{Bre18} with default settings, except that \texttt{ntree} is set to 1000. At each stage, predictors are considered significant if they rank among the top third in importance as determined by the \texttt{randomForest} function.
\item[iii)]
AOWL is implemented using the \texttt{DTRlearn2} R package with default settings, except that \texttt{loss} is set to \texttt{l2} and \texttt{augment} is set to \texttt{TRUE}, following the package guidelines.
\end{itemize}

\subsection{Simulation results for the regime $n>d_2$}\label{sec:n_geq_d2}

The next tables report the results of the same analysis displayed in Section \ref{PART1}, this time for data simulated under the regime with $n>d_2$.

\begin{table}[H]
\caption{First experiment: simulated data with $n$ larger than $d_2$. For stage 2 and stage 1, three summaries of the accuracy in selecting significant variables (proportion of FN, proportion of FP, and F$_1$ score), for DSS, ISS, Q-learning with lasso (QL), and Q-learning with RF (QRF), across the nine scenarios obtained by specifying $(k,n)\in\{(10,50),(20,100),(30,150)\}$  and $\rho^*\in\{0.3,0.6,0.9\}$. Results are based on 100 replicated datasets. }
\centering
\resizebox{1\textwidth}{!}{
\begin{tabular}{cccccccccccccccccccccc}
\Hline
\multicolumn{21}{c}{Stage 2}                                                                                                                         \\ \hline
\multirow{2}{*}{$n$} & \multirow{2}{*}{$d_2$} & \multirow{2}{*}{$\rho^*$} &  & \multicolumn{4}{c}{FN} &  & \multicolumn{4}{c}{FP} &  & \multicolumn{4}{c}{F$_1$ score} \\ \cline{5-8} \cline{10-13} \cline{15-18}  
&  &  & & DSS    & ISS   & QL  &QRF  &  & DSS    & ISS   & QL  &QRF  &  & DSS      & ISS     & QL   &QRF    \\ \hline 
         &                          & 0.3 &   &	0.018	&	0.008	&	0.027	&	0.401 &	&	0.117	&	0.048	&	0.217	&	0.172	& &	0.896	&	0.954	&	0.804	&	0.594	\\
         &                          & 0.6 &   &	0.006	&	0.008	&	0.026	&	0.391&	&	0.071	&	0.049	&	0.225	&	0.169&	&	0.935	&	0.951	&	0.806	&	0.601	\\
\multirow{-3}{*}{50}            & \multirow{-3}{*}{33} & 0.9 &  &	0.004	&	0.010	&	0.029	&	0.377	& &	0.042	&	0.047	&	0.244	&	0.165 &	&	0.962	&	0.955	&	0.790	&	0.612	\\ [5pt]
         &                          & 0.3 &   &	0.008	&	0.005	&	0.010	&	0.395 &	&	0.041	&	0.021	&	0.185	&	0.179 &	&	0.954	&	0.975	&	0.825	&	0.587	\\
         &                          & 0.6 &   &		0.005	&	0.003	&	0.014	&	0.399&	&	0.033	&	0.019	&	0.180	&	0.182	 & &	0.962	&	0.977	&	0.822	&	0.597	\\
\multirow{-3}{*}{100}           & \multirow{-3}{*}{63}& 0.9 &  &	0.002	&	0.006	&	0.013	&	0.385	& &	0.021	&	0.026	&	0.192	&	0.175 &	&	0.977	&	0.971	&	0.820	&	0.569	\\ [5pt]
        &                          & 0.3 &   &		0.003	&	0.003	&	0.007	&	0.426	& &	0.030	&	0.019	&	0.184	&	0.185	& &	0.968	&	0.979	&	0.828	&	0.569	\\
         &                          & 0.6 &   &	0.004	&	0.005	&	0.010	&	0.440&	&	0.022	&	0.017	&	0.181	&	0.189 &	&	0.975	&	0.980	&	0.830	&	0.559	\\
\multirow{-3}{*}{150}           & \multirow{-3}{*}{93} & 0.9 &  &	0.001	&	0.003	&	0.007	&	0.430	& &	0.012	&	0.017	&	0.173	&	0.189	&	&0.986	&	0.981	&	0.837	&	0.567	\\  \hline
\multicolumn{21}{c}{Stage 1}                                                                                                                               \\ \hline
\multirow{2}{*}{$n$} & \multirow{2}{*}{$d_1$} & \multirow{2}{*}{$\rho^*$} &  & \multicolumn{4}{c}{FN} &  & \multicolumn{4}{c}{FP} &  & \multicolumn{4}{c}{F$_1$ score} \\ \cline{5-8} \cline{10-13} \cline{15-18} 
 &  &  & & DSS    & ISS   & QL  &QRF &  & DSS    & ISS   & QL  &QRF &  & DSS      & ISS     & QL   &QRF       \\ \hline
         &                          & 0.3 &   &	0.018	&	0.015	&	0.036	&	0.324 &	&	0.225	&	0.011	&	0.177	&	0.172	 & &	0.783	&	0.982	&	0.825	&	0.623	\\
         &                          & 0.6 &   &	0.006	&	0.010	&	0.019	&	0.307 &	&	0.132	&	0.017	&	0.182	&	0.159 &	&	0.878	&	0.981	&	0.838	&	0.648	\\
\multirow{-3}{*}{50}            & \multirow{-3}{*}{22}& 0.9 &  & 0.004	&	0.025	&	0.051	&	0.348 &	&	0.035	&	0.015	&	0.194	&	0.165	& &	0.967	&	0.973	&	0.822	&	0.617	\\ [5pt]
         &                          & 0.3 &   &	0.008	&	0.014	&	0.014	&	0.355 &	&	0.103	&	0.002	&	0.135	&	0.179 &	&	0.895	&	0.990	&	0.869	&	0.613	\\
         &                          & 0.6 &   &	0.005	&	0.012	&	0.016	&	0.345 &	&	0.065	&	0.004	&	0.116	&	0.182 &	&	0.937	&	0.990	&	0.881	&	0.625	\\
\multirow{-3}{*}{100}           & \multirow{-3}{*}{42}& 0.9 &  &	0.002	&	0.017	&	0.018	&	0.342	& &	0.021	&	0.004	&	0.110	&	0.175	& &	0.978	&	0.988	&	0.882	&	0.622	\\ [5pt]
         &                          & 0.3 &   &	0.003	&	0.003	&	0.005	&	0.360 &	&	0.065	&	0.002	&	0.124	&	0.168 &	&	0.936	&	0.996	&	0.885	&	0.620	\\
         &                          & 0.6 &   &	0.004	&	0.008	&	0.010	&	0.377 &	&	0.036	&	0.001	&	0.118	&	0.175 &	&	0.962	&	0.995	&	0.881	&	0.604	\\
\multirow{-3}{*}{150}           & \multirow{-3}{*}{62}& 0.9 &  &	0.001	&	0.007	&	0.010	&	0.379 &	&	0.012	&	0.000	&	0.108	&	0.172	& &	0.988	&	0.996	&	0.892	&	0.608	\\ \hline
\end{tabular}
}
\label{TABLE first experiment_1_1}
\end{table}

\begin{table}[H]
\centering
\caption{First experiment: simulated data with $n$ larger than $d_2$. For stage 2, stage 1 and overall, two measures of prediction accuracy (MAE and ER), for DSS, ISS, Q-learning with lasso (QL), Q-learning with RF (QRF) and AOWL, across the nine scenarios obtained by specifying $(k,n)\in\{(10,50),(20,100),(30,150)\}$ and $\rho^*\in\{0.3,0.6,0.9\}$. Results are based on 100 replicated datasets.}
\resizebox{1\textwidth}{!}{
\begin{tabular}{ccccccccccccccc}
\Hline
\multicolumn{15}{c}{Stage 2}                                                                                \\ \hline
         \multirow{2}{*}{$n$} &   \multirow{2}{*}{$d_2$}                               &   \multirow{2}{*}{$\rho^*$}   & & \multicolumn{5}{c}{MAE}                                            &                              & \multicolumn{5}{c}{ER}                                                           \\ \cline{5-9} \cline{11-15} 
 &   & & & DSS                             & ISS                           & QL     &QRF & AOWL  &                    & DSS                           & ISS                           & QL    &QRF & AOWL                        \\ \hline
\multicolumn{1}{c}{}                                   & \multicolumn{1}{c}{}             &	0.3	&	&	0.018	&	0.018	&	0.045	&	0.647	&	0.474	&	&	0.041	&	0.042	&	0.054	&	0.234	&	0.197	\\	
\multicolumn{1}{c}{}                                   & \multicolumn{1}{c}{}  &	0.6	&	&	0.018	&	0.020	&	0.033	&	0.589	&	0.456	&	&	0.045	&	0.047	&	0.052	&	0.225	&	0.191	\\	
\multicolumn{1}{c}{\multirow{-3}{*}{50}}               & \multicolumn{1}{c}{\multirow{-3}{*}{33}}     &	0.9	&	&	0.014	&	0.015	&	0.032	&	0.621	&	0.470	&	&	0.029	&	0.029	&	0.041	&	0.227	&	0.188	\\	[5pt]
\multicolumn{1}{c}{}                                   & \multicolumn{1}{c}{} &	0.3	&	&	0.009	&	0.009	&	0.014	&	0.955	&	0.571	&	&	0.020	&	0.019	&	0.025	&	0.249	&	0.192	\\	
\multicolumn{1}{c}{}                                   & \multicolumn{1}{c}{} &	0.6	&	&	0.008	&	0.007	&	0.019	&	0.928	&	0.556	&	&	0.023	&	0.022	&	0.030	&	0.250	&	0.187	\\	
\multicolumn{1}{c}{\multirow{-3}{*}{100}}               & \multicolumn{1}{c}{\multirow{-3}{*}{63}}  & 	0.9	&	&	0.007	&	0.008	&	0.014	&	0.960	&	0.544	&	&	0.025	&	0.026	&	0.031	&	0.248	&	0.184	\\	[5pt]
\multicolumn{1}{c}{}                                   & \multicolumn{1}{c}{}             &	0.3	&	&	0.008	&	0.007	&	0.009	&	1.283	&	0.669	&	&	0.020	&	0.020	&	0.021	&	0.260	&	0.192	\\	
\multicolumn{1}{c}{}                                   & \multicolumn{1}{c}{}             &	0.6	&	&	0.006	&	0.007	&	0.011	&	1.341	&	0.666	&	&	0.018	&	0.019	&	0.023	&	0.275	&	0.191	\\	
\multicolumn{1}{c}{\multirow{-3}{*}{150}}               & \multicolumn{1}{c}{\multirow{-3}{*}{93}} &	0.9	&	&	0.006	&	0.006	&	0.009	&	1.329	&	0.704	&	&	0.017	&	0.017	&	0.021	&	0.263	&	0.188	\\	\hline
\multicolumn{15}{c}{Stage 1} 
\\
\hline
\multirow{2}{*}{$n$} & \multirow{2}{*}{$d_1$} & \multirow{2}{*}{$\rho^*$} & & \multicolumn{5}{c}{MAE}                                       &                                   & \multicolumn{5}{c}{ER}                                                                        \\ \cline{5-9} \cline{11-15} 
 &   & & & DSS                             & ISS                           & QL     &QRF & AOWL  &                    & DSS                           & ISS                           & QL    &QRF & AOWL                                            \\ \hline
\multicolumn{1}{c}{}                                   & \multicolumn{1}{c}{}             &	0.3	&	&	0.021	&	0.020	&	0.035	&	0.579	&	0.371	&	&	0.059	&	0.057	&	0.067	&	0.236	&	0.178	\\	
\multicolumn{1}{c}{}                                   & \multicolumn{1}{c}{}  &	0.6	&	&	0.018	&	0.022	&	0.035	&	0.492	&	0.356	&	&	0.056	&	0.060	&	0.064	&	0.209	&	0.180	\\	
\multicolumn{1}{c}{\multirow{-3}{*}{50}}               & \multicolumn{1}{c}{\multirow{-3}{*}{22}}     &	0.9	&	&	0.025	&	0.034	&	0.049	&	0.553	&	0.368	&	&	0.045	&	0.052	&	0.066	&	0.209	&	0.164	\\	[5pt]
\multicolumn{1}{c}{}                                   & \multicolumn{1}{c}{} &	0.3	&	&	0.012	&	0.010	&	0.015	&	0.922	&	0.465	&	&	0.024	&	0.023	&	0.027	&	0.242	&	0.168	\\	
\multicolumn{1}{c}{}                                   & \multicolumn{1}{c}{} &	0.6	&	&	0.013	&	0.012	&	0.023	&	0.931	&	0.445	&	&	0.033	&	0.035	&	0.044	&	0.251	&	0.179	\\	
\multicolumn{1}{c}{\multirow{-3}{*}{100}}               & \multicolumn{1}{c}{\multirow{-3}{*}{42}}  & 	0.9	&	&	0.009	&	0.013	&	0.018	&	0.976	&	0.482	&	&	0.026	&	0.028	&	0.032	&	0.259	&	0.181	\\	[5pt]
\multicolumn{1}{c}{}                                   & \multicolumn{1}{c}{}             &	0.3	&	&	0.011	&	0.008	&	0.011	&	1.269	&	0.576	&	&	0.024	&	0.021	&	0.024	&	0.259	&	0.172	\\	
\multicolumn{1}{c}{}                                   & \multicolumn{1}{c}{}             &	0.6	&	&	0.010	&	0.010	&	0.015	&	1.285	&	0.514	&	&	0.023	&	0.022	&	0.027	&	0.266	&	0.166	\\	
\multicolumn{1}{c}{\multirow{-3}{*}{150}}               & \multicolumn{1}{c}{\multirow{-3}{*}{62}} &	0.9	&	&	0.008	&	0.009	&	0.013	&	1.298	&	0.541	&	&	0.020	&	0.021	&	0.026	&	0.268	&	0.166	\\	\hline

\multicolumn{15}{c}{Overall}                                                             \\ \hline
\multirow{2}{*}{$n$} & \multirow{2}{*}{$(d_1,d_2)$} & \multirow{2}{*}{$\rho^*$} & &  \multicolumn{5}{c}{MAE}                                        &                                  & \multicolumn{5}{c}{ER}                                                                        \\ \cline{5-9} \cline{11-15}                                          &   & & & DSS                             & ISS                           & QL     &QRF & AOWL  &                    & DSS                           & ISS                           & QL    &QRF & AOWL                                        \\ \hline
\multicolumn{1}{c}{}                                   & \multicolumn{1}{c}{}             &	0.3	&	&	0.040	&	0.039	&	0.080	&	1.227	&	0.845	&	&	0.098	&	0.096	&	0.119	&	0.415	&	0.339	\\	
\multicolumn{1}{c}{}                                   & \multicolumn{1}{c}{}  &	0.6	&	&	0.037	&	0.042	&	0.068	&	1.081	&	0.812	&	&	0.100	&	0.103	&	0.113	&	0.390	&	0.341	\\	
\multicolumn{1}{c}{\multirow{-3}{*}{50}}               & \multicolumn{1}{c}{\multirow{-3}{*}{(22, 33)}}     &	0.9	&	&	0.039	&	0.049	&	0.081	&	1.174	&	0.837	&	&	0.073	&	0.081	&	0.105	&	0.393	&	0.318	\\	[5pt]
\multicolumn{1}{c}{}                                   & \multicolumn{1}{c}{} &	0.3	&	&	0.021	&	0.018	&	0.029	&	1.877	&	1.036	&	&	0.044	&	0.042	&	0.051	&	0.433	&	0.328	\\	
\multicolumn{1}{c}{}                                   & \multicolumn{1}{c}{} &	0.6	&	&	0.021	&	0.020	&	0.042	&	1.859	&	1.001	&	&	0.056	&	0.057	&	0.074	&	0.438	&	0.333	\\	
\multicolumn{1}{c}{\multirow{-3}{*}{100}}               & \multicolumn{1}{c}{\multirow{-3}{*}{(42, 63)}}  & 	0.9	&	&	0.016	&	0.021	&	0.032	&	1.936	&	1.026	&	&	0.051	&	0.055	&	0.061	&	0.443	&	0.331	\\	[5pt]
\multicolumn{1}{c}{}                                   & \multicolumn{1}{c}{}             &	0.3	&	&	0.018	&	0.015	&	0.019	&	2.552	&	1.245	&	&	0.043	&	0.040	&	0.045	&	0.450	&	0.331	\\	
\multicolumn{1}{c}{}                                   & \multicolumn{1}{c}{}             &	0.6	&	&	0.016	&	0.017	&	0.026	&	2.626	&	1.180	&	&	0.041	&	0.041	&	0.049	&	0.467	&	0.323	\\	
\multicolumn{1}{c}{\multirow{-3}{*}{150}}               & \multicolumn{1}{c}{\multirow{-3}{*}{(62, 93)}} &	0.9	&	&	0.014	&	0.015	&	0.021	&	2.627	&	1.245	&	&	0.036	&	0.037	&	0.046	&	0.462	&	0.321	\\	\hline
\end{tabular}}
\label{TABLE first experiment_2_2}
\end{table}

\subsection{Sensitivity analysis on $r$}\label{sec:sensitivity}
The ratio $r$ between the variances of the spike and slab components in the dependent spike-and-slab model plays a central role in enabling variable selection across stages. To investigate its impact, we conducted a simple sensitivity analysis across different values of $r$, which regulates the extent of prior shrinkage. The results, based on a specific scenario from the first simulation study, are reported in Table \ref{tab:rho_r_results}. Results demonstrate that our method remains robust across a range of $r$ values, as indicated by the corresponding $\text{F}_1$ scores and ER.
\begin{table}[H]
\centering
\caption{First experiment, sensitivity analysis. For stage 2 and stage 1, F$_1$ score and ER, for DSS with $k=10$, $n=25$  and $\rho^*\in\{0.3,0.6,0.9\}$, and different values of $r\in\{0.01,0.001,0.0001\}$. Results are based on 100 replicated datasets. }
\label{tab:rho_r_results}
\begin{tabular}{clccccc}
\Hline
&  & \multicolumn{2}{c}{Stage 2} & & \multicolumn{2}{c}{Stage 1} \\
\cline{3-4} \cline{6-7}
 \multirow{1}{*}{$\rho^*$}  & \multirow{1}{*}{$r$} & $\text{F}_1$ score & ER & & $\text{F}_1$ score & ER \\
\hline
\multirow{3}{*}{0.3}  & 0.01   & 0.711 & 0.184 & & 0.627 & 0.142 \\
     & 0.001  & 0.704 & 0.136 & & 0.626 & 0.183 \\
     & 0.0001 & 0.704 & 0.135 & & 0.620 & 0.183 \\
\hline
\multirow{3}{*}{0.6}  & 0.01   & 0.747 & 0.128 & & 0.720 & 0.176 \\
     & 0.001  & 0.744 & 0.130 & & 0.719 & 0.179 \\
     & 0.0001 & 0.741 & 0.121 & & 0.706 & 0.181 \\
\hline
\multirow{3}{*}{0.9}  & 0.01   & 0.786 & 0.152 & & 0.810 & 0.168 \\
     & 0.001  & 0.770 & 0.148 & & 0.806 & 0.175 \\
     & 0.0001 & 0.765 & 0.152 & & 0.786 & 0.177 \\
\hline
\end{tabular}
\end{table}

\subsection{Experiment beyond linearity}\label{sec:nonlin}
We report the results of an additional simulation experiment in which the data are generated from a process where the relationship between the stage-specific expected payoffs and half of the predictors is linear, while the effect of the remaining predictors is nonlinear. Specifically, the data are generated according to the process described in Section \ref{sec:dgp3}.\\ 
We begin by considering the same scenarios as in the first experiment under the regime $n > d_2$ (see Section \ref{sec:n_geq_d2}), namely $(k,n)\in\{(10,50),(20,100),(30,150)\}$. The results are presented in Tables \ref{TABLE third experiment 1} and \ref{TABLE third experiment 2}. Compared with the corresponding results in Tables \ref{TABLE first experiment_1_1} and \ref{TABLE first experiment_2_2}, which are based on data generated from a linear process, the partially linear data generating process appears substantially more challenging. This is particularly evident for methods such as DSS, ISS, and Q-learning with lasso, which rely on linearity assumptions and exhibit poorer performance in terms of metrics such as the F$_1$-score and the error rate. In contrast, these same metrics indicate that Q-learning with RF is less affected by the nonlinearity of the data generating process, making it a competitive alternative in this setting. A similar conclusion applies to AOWL in terms of predictive accuracy.

\begin{table}[H]
\caption{Third experiment: data simulated from a partially linear data generating process. For stage 2 and stage 1, three summaries of the accuracy in selecting significant variables (proportion of FN, proportion of FP, and F$_1$ score), for DSS, ISS, Q-learning with lasso (QL), and Q-learning with RF (QRF), across the nine scenarios obtained by specifying $(k,n)\in\{(10,50),(20,100),(30,150)\}$  and $\rho^*\in\{0.3,0.6,0.9\}$. Results are based on 100 replicated datasets.}
\centering
\resizebox{1\textwidth}{!}{
\begin{tabular}{cccccccccccccccccccccc}
\Hline
\multicolumn{21}{c}{Stage 2}\\ \hline
\multirow{2}{*}{$n$} & \multirow{2}{*}{$d_2$} & \multirow{2}{*}{$\rho^*$} & & \multicolumn{4}{c}{FN} & & \multicolumn{4}{c}{FP} & & \multicolumn{4}{c}{F$_1$ score} \\ \cline{5-8} \cline{10-13} \cline{15-18} 
& & & & DSS & ISS  & QL &QRF & & DSS & ISS & QL &QRF & & DSS  & ISS & QL &QRF \\ \hline
 & & 0.3 & & 0.322 & 0.346 & 0.389 & 0.432 & & 0.141 & 0.115 & 0.184 & 0.195 & & 0.666 & 0.672 & 0.574 & 0.548 \\
 & & 0.6 & & 0.308 & 0.351 & 0.404 & 0.410& & 0.153 & 0.123 & 0.176 & 0.189& & 0.664 & 0.660 & 0.561 & 0.561 \\
\multirow{-3}{*}{50} & \multirow{-3}{*}{33} & 0.9 & & 0.313 & 0.374 & 0.480 & 0.426 & & 0.129 & 0.114 & 0.164 & 0.190 & & 0.694 & 0.657 & 0.523 & 0.555 \\ [5pt]
 & & 0.3 & & 0.296 & 0.312 & 0.369 & 0.453 & & 0.168 & 0.139 & 0.172 & 0.200 & & 0.671 & 0.685 & 0.606 & 0.536 \\
 & & 0.6 & & 0.262 & 0.310 & 0.354 & 0.434& & 0.158 & 0.126 & 0.190 & 0.195& & 0.698 & 0.689 & 0.601 & 0.548 \\
\multirow{-3}{*}{100} & \multirow{-3}{*}{63}& 0.9 & & 0.231 & 0.282 & 0.375 & 0.435 & & 0.122 & 0.120 & 0.160 & 0.198 & & 0.747 & 0.713 & 0.603 & 0.544 \\ [5pt]
& & 0.3 & & 0.251 & 0.274 & 0.332 & 0.453 & & 0.158 & 0.121 & 0.176 & 0.199 & & 0.705 & 0.720 & 0.635 & 0.537 \\
&  & 0.6 & & 0.223 & 0.262 & 0.331 & 0.464 & & 0.164 & 0.142 & 0.178 & 0.200& & 0.725 & 0.718 & 0.640 & 0.533 \\
\multirow{-3}{*}{150} & \multirow{-3}{*}{93} & 0.9 & & 0.214 & 0.274 & 0.336 & 0.470 & & 0.136 & 0.134 & 0.168 & 0.207 & &0.750 & 0.713 & 0.637 & 0.521 \\ \hline
\multicolumn{21}{c}{Stage 1}  \\ \hline
\multirow{2}{*}{$n$} & \multirow{2}{*}{$d_1$} & \multirow{2}{*}{$\rho^*$} & & \multicolumn{4}{c}{FN} & & \multicolumn{4}{c}{FP} & & \multicolumn{4}{c}{F$_1$ score} \\ \cline{5-8} \cline{10-13} \cline{15-18}
& & & & DSS & ISS & QL &QRF & & DSS & ISS & QL &QRF & & DSS & ISS & QL &QRF \\ \hline
 & & 0.3 & & 0.399 & 0.534 & 0.580 & 0.387 & & 0.205 & 0.063 & 0.139 & 0.212 & & 0.546 & 0.543 & 0.420 & 0.538 \\
 & & 0.6 & & 0.317 & 0.520 & 0.543 & 0.340 & & 0.163 & 0.073 & 0.119 & 0.185 & & 0.646 & 0.558 & 0.472 & 0.603 \\
\multirow{-3}{*}{50} & \multirow{-3}{*}{22}& 0.9 & & 0.284 & 0.535 & 0.552 & 0.397 & & 0.134 & 0.074 & 0.148 & 0.195 & & 0.707 & 0.551 & 0.462 & 0.569 \\ [5pt]
 & & 0.3 & & 0.349 & 0.448 & 0.498 & 0.405 & & 0.231 & 0.086 & 0.113 & 0.192 & & 0.591 & 0.624 & 0.536 & 0.570 \\
& & 0.6 & & 0.304 & 0.457 & 0.466 & 0.386 & & 0.174 & 0.081 & 0.122 & 0.188 & & 0.656 & 0.617 & 0.564 & 0.582 \\
\multirow{-3}{*}{100} & \multirow{-3}{*}{42}& 0.9 & & 0.219 & 0.473 & 0.466 & 0.375 & & 0.122 & 0.083 & 0.129 & 0.180 & & 0.752 & 0.601 & 0.553 & 0.597 \\ [5pt]
& & 0.3 & & 0.345 & 0.497 & 0.497 & 0.438 & & 0.230 & 0.095 & 0.117 & 0.197 & & 0.597 & 0.625 & 0.543 & 0.552 \\
&  & 0.6 & & 0.257 & 0.475 & 0.476 & 0.422 & & 0.196 & 0.103 & 0.102 & 0.195 & & 0.673 & 0.643 & 0.568 & 0.557 \\
\multirow{-3}{*}{150} & \multirow{-3}{*}{62}& 0.9 & & 0.182 & 0.409 & 0.479 & 0.421 & & 0.129 & 0.103 & 0.102 & 0.197 & & 0.769 & 0.645 & 0.562 & 0.554 \\ \hline
\end{tabular}
}
\label{TABLE third experiment 1}
\end{table}

\begin{table}[H]
\centering
\caption{Third experiment: data simulated from a partially linear data generating process. For stage 2, stage 1 and overall, two measures of prediction accuracy (MAE and ER), for DSS, ISS, Q-learning with lasso (QL), Q-learning with RF (QRF) and AOWL, across the nine scenarios obtained by specifying $(k,n)\in\{(10,50),(20,100),(30,150)\}$  and $\rho^*\in\{0.3,0.6,0.9\}$. Results are based on 100 replicated datasets.}
\resizebox{1\textwidth}{!}{
\begin{tabular}{ccccccccccccccc}
\Hline
\multicolumn{15}{c}{Stage 2} \\ \hline
 \multirow{2}{*}{$n$} &  \multirow{2}{*}{$d_2$} & \multirow{2}{*}{$\rho^*$} & & \multicolumn{5}{c}{MAE} & & \multicolumn{5}{c}{ER} \\ \cline{5-9} \cline{11-15} & & & & DSS & ISS & QL & QRF & AOWL & & DSS & ISS & QL &QRF & AOWL \\ \hline
\multicolumn{1}{c}{}  & \multicolumn{1}{c}{}  & 0.3 & & 4.458 & 4.472 & 4.628 & 4.512 & 4.473 & & 0.203 & 0.201 & 0.271 & 0.215 & 0.249 \\
\multicolumn{1}{c}{} & \multicolumn{1}{c}{} & 0.6 & & 4.280 & 4.284 & 4.469 & 4.317 & 4.287 & & 0.196 & 0.201 & 0.285 & 0.226 & 0.253 \\
\multicolumn{1}{c}{\multirow{-3}{*}{50}} & \multicolumn{1}{c}{\multirow{-3}{*}{33}} & 0.9 & & 4.413 & 4.441 & 4.624 & 4.466 & 4.517 & & 0.190 & 0.194 & 0.276 & 0.212 & 0.249 \\ [5pt]
\multicolumn{1}{c}{} & \multicolumn{1}{c}{} & 0.3 & & 6.585 & 6.576 & 6.588 & 6.554 & 6.639 & & 0.145 & 0.144 & 0.172 & 0.184 & 0.203 \\
\multicolumn{1}{c}{} & \multicolumn{1}{c}{} & 0.6 & & 6.543 & 6.543 & 6.615 & 6.605 & 6.563 & & 0.158 & 0.163 & 0.176 & 0.194 & 0.217 \\
\multicolumn{1}{c}{\multirow{-3}{*}{100}} & \multicolumn{1}{c}{\multirow{-3}{*}{63}} & 0.9 & & 6.179 & 6.172 & 6.275 & 6.285 & 6.303 & & 0.151 & 0.153 & 0.174 & 0.193 & 0.212 \\ [5pt]
\multicolumn{1}{c}{} & \multicolumn{1}{c}{} & 0.3 & & 8.014 & 7.990 & 8.034 & 8.139 & 8.097 & & 0.171 & 0.173 & 0.194 & 0.219 & 0.222 \\
\multicolumn{1}{c}{} & \multicolumn{1}{c}{} & 0.6 & & 8.163 & 8.176 & 8.262 & 8.265 & 8.258 & & 0.146 & 0.146 & 0.162 & 0.193 & 0.198 \\
\multicolumn{1}{c}{\multirow{-3}{*}{150}} & \multicolumn{1}{c}{\multirow{-3}{*}{93}} & 0.9 & & 8.391 & 8.386 & 8.411 & 8.451 & 8.445 & & 0.151 & 0.153 & 0.169 & 0.200 & 0.194 \\ \hline
\multicolumn{15}{c}{Stage 1} 
\\
\hline
\multirow{2}{*}{$n$} & \multirow{2}{*}{$d_1$} & \multirow{2}{*}{$\rho^*$} & & \multicolumn{5}{c}{MAE}                                       &                                   & \multicolumn{5}{c}{ER}                                                                        \\ \cline{5-9} \cline{11-15} 
 &   & & & DSS                             & ISS                           & QL     &QRF & AOWL  &                    & DSS                           & ISS                           & QL    &QRF & AOWL                                            \\ \hline
\multicolumn{1}{c}{}                                   & \multicolumn{1}{c}{}             & 0.3 & & 2.565 & 2.580 & 2.920 & 2.598 & 2.579 & & 0.228 & 0.238 & 0.383 & 0.234 & 0.256 \\
\multicolumn{1}{c}{}                                   & \multicolumn{1}{c}{}  & 0.6 & & 2.578 & 2.579 & 2.944 & 2.570 & 2.628 & & 0.225 & 0.231 & 0.390 & 0.217 & 0.247 \\
\multicolumn{1}{c}{\multirow{-3}{*}{50}}               & \multicolumn{1}{c}{\multirow{-3}{*}{22}}     & 0.9 & & 2.707 & 2.741 & 3.046 & 2.741 & 2.793 & & 0.192 & 0.204 & 0.362 & 0.206 & 0.248 \\ [5pt]
\multicolumn{1}{c}{}                                   & \multicolumn{1}{c}{} & 0.3 & & 3.892 & 3.890 & 4.078 & 3.919 & 3.914 & & 0.195 & 0.202 & 0.259 & 0.218 & 0.219 \\
\multicolumn{1}{c}{}                                   & \multicolumn{1}{c}{} & 0.6 & & 4.009 & 4.008 & 4.128 & 4.058 & 4.059 & & 0.176 & 0.182 & 0.232 & 0.195 & 0.209 \\
\multicolumn{1}{c}{\multirow{-3}{*}{100}}               & \multicolumn{1}{c}{\multirow{-3}{*}{42}}  & 0.9 & & 4.081 & 4.113 & 4.287 & 4.125 & 4.100 & & 0.204 & 0.217 & 0.253 & 0.237 & 0.239 \\ [5pt]
\multicolumn{1}{c}{}                                   & \multicolumn{1}{c}{}             & 0.3 & & 4.811 & 4.806 & 4.986 & 4.869 & 4.896 & & 0.181 & 0.183 & 0.219 & 0.210 & 0.207 \\
\multicolumn{1}{c}{}                                   & \multicolumn{1}{c}{}             & 0.6 & & 4.766 & 4.772 & 4.846 & 4.787 & 4.851 & & 0.165 & 0.174 & 0.204 & 0.199 & 0.197 \\
\multicolumn{1}{c}{\multirow{-3}{*}{150}}               & \multicolumn{1}{c}{\multirow{-3}{*}{62}} & 0.9 & & 4.979 & 5.029 & 5.149 & 5.083 & 5.071 & & 0.161 & 0.174 & 0.206 & 0.201 & 0.200 \\ \hline

\multicolumn{15}{c}{Overall} \\ \hline
\multirow{2}{*}{$n$} & \multirow{2}{*}{$(d_1,d_2)$} & \multirow{2}{*}{$\rho^*$} & & \multicolumn{5}{c}{MAE} & & \multicolumn{5}{c}{ER} \\ \cline{5-9} \cline{11-15} & & & & DSS & ISS & QL & QRF & AOWL & & DSS & ISS & QL &QRF & AOWL \\ \hline
\multicolumn{1}{c}{} & \multicolumn{1}{c}{} & 0.3 & & 5.389 & 5.436 & 5.786 & 5.449 & 5.425 & & 0.385 & 0.392 & 0.549 & 0.404 & 0.443 \\
\multicolumn{1}{c}{} & \multicolumn{1}{c}{} & 0.6 & & 5.084 & 5.110 & 5.532 & 5.112 & 5.137 & & 0.379 & 0.386 & 0.566 & 0.395 & 0.434 \\
\multicolumn{1}{c}{\multirow{-3}{*}{50}} & \multicolumn{1}{c}{\multirow{-3}{*}{(22, 33)}} & 0.9 & & 5.137 & 5.176 & 5.523 & 5.216 & 5.288 & & 0.346 & 0.358 & 0.530 & 0.375 & 0.432 \\ [5pt]
\multicolumn{1}{c}{} & \multicolumn{1}{c}{} & 0.3 & & 7.805 & 7.789 & 7.911 & 7.796 & 7.903 & & 0.310 & 0.316 & 0.385 & 0.362 & 0.376 \\
\multicolumn{1}{c}{} & \multicolumn{1}{c}{} & 0.6 & & 7.630 & 7.615 & 7.800 & 7.757 & 7.682 & & 0.303 & 0.313 & 0.367 & 0.348 & 0.377 \\
\multicolumn{1}{c}{\multirow{-3}{*}{100}} & \multicolumn{1}{c}{\multirow{-3}{*}{(42, 63)}} & 0.9 & & 7.450 & 7.473 & 7.699 & 7.616 & 7.531 & & 0.326 & 0.337 & 0.382 & 0.386 & 0.399 \\ [5pt]
\multicolumn{1}{c}{} & \multicolumn{1}{c}{} & 0.3 & & 9.391 & 9.350 & 9.516 & 9.487 & 9.523 & & 0.320 & 0.325 & 0.368 & 0.381 & 0.381 \\
\multicolumn{1}{c}{} & \multicolumn{1}{c}{} & 0.6 & & 9.597 & 9.604 & 9.741 & 9.731 & 9.733 & & 0.288 & 0.295 & 0.334 & 0.358 & 0.357 \\
\multicolumn{1}{c}{\multirow{-3}{*}{150}} & \multicolumn{1}{c}{\multirow{-3}{*}{(62, 93)}} & 0.9 & & 9.789 & 9.805 & 9.934 & 9.907 & 9.859 & & 0.290 & 0.301 & 0.342 & 0.361 & 0.356 \\ \hline
\end{tabular}}
\label{TABLE third experiment 2}
\end{table}

We also examined how sample size affects the performance of the considered methods when data are generated using this more complex data generating process. Specifically, we set a medium level of correlation, $\rho^*=0.6$, with $k=10$, resulting in $d_1=22$ and $d_2=33$, and considered sample sizes $n \in \{25,50,100,200,400\}$. Figures~\ref{fig:F1} focus on the stage-wise F$_1$-score as a summary of the accuracy in selecting significant variables. It can be observed that when $n=25$, i.e., $n<d_2$, methods relying on linearity assumptions, which are violated here, perform poorly, whereas Q-learning with RF is less affected. As $n$ increases, the accuracy of the former methods improves more rapidly than that of Q-learning with RF, which does not appear to be particularly effective for variable selection in large samples.

\begin{figure}[htbp]
    \centering
    
    \includegraphics[width=\textwidth]{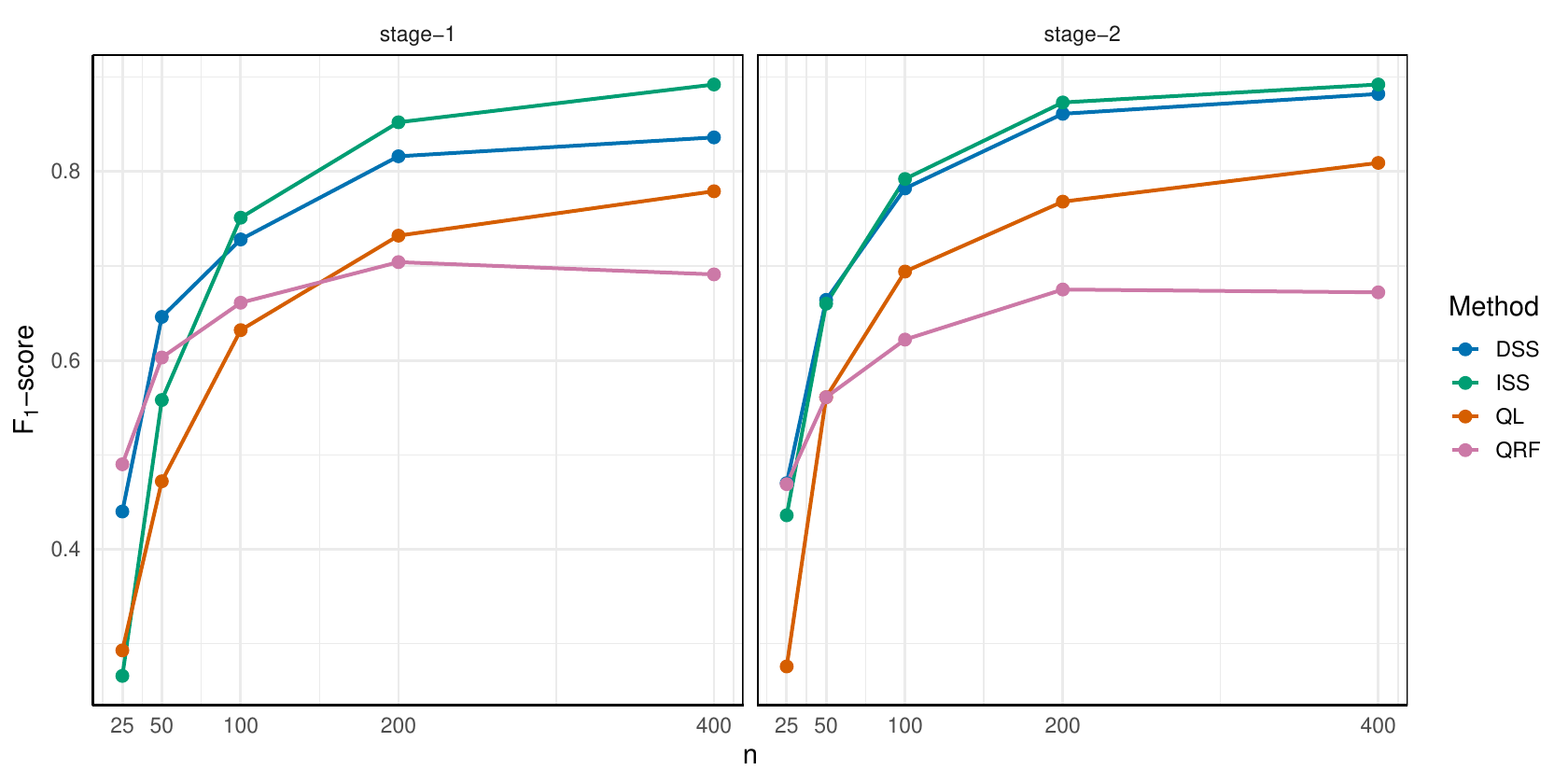}
    
    \caption{Third experiment: data simulated from a partially linear data-generating process. Stage-wise F$_1$-scores, summarizing the accuracy in selecting significant variables, are shown as a function of sample size $n \in \{25,50,100,200,400\}$ for DSS, ISS, Q-learning with lasso (QL), and Q-learning with RF (QRF). Parameters were set to $k=10$ and $\rho^*=0.6$. Results are based on 100 replicated datasets.}
    \label{fig:F1}
\end{figure}

A different aspect is captured by the error rate, which measures prediction error. Figure~\ref{fig:ER} shows that, in these terms, Q-learning with RF and AOWL remain highly competitive across all considered sample sizes, with performance very similar to the two BAL options.

\begin{figure}[htbp]
    \centering
    
    \includegraphics[width=0.675\textwidth]{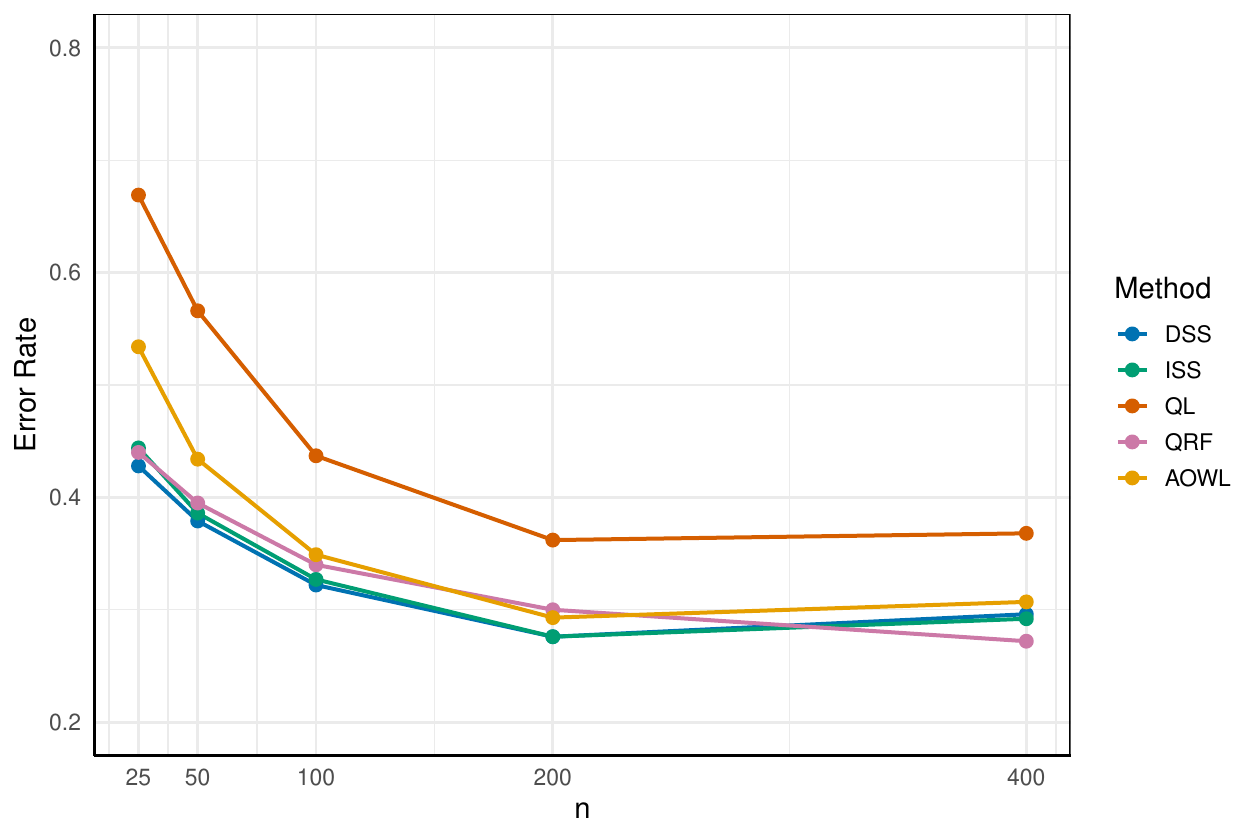}
    
    \caption{Third experiment: data simulated from a partially linear data-generating process. The error rate, measure of overall prediction accuracy, is shown as a function of sample size $n \in \{25,50,100,200,400\}$ for DSS, ISS, Q-learning with lasso (QL), Q-learning with RF (QRF), and AOWL. Parameters were set to $k=10$ and $\rho^*=0.6$. Results are based on 100 replicated datasets.}
    \label{fig:ER}
\end{figure}

\subsection{Violation of the positivity assumption}\label{sec:appPOS}
To assess our method under violations of the positivity assumption \citep[see, e.g.,][]{Wes10}, we designed as a simulation setting a minor variation of the data generating process considered for the second simulation experiment in Section \ref{PART2}. In this variation, the stage-1 treatment allocation is intentionally made highly imbalanced. Two of the four treatment options are assigned a probability of 0.01, while the other two receive a probability of 0.49, with a sample size $n=200$. In contrast, the stage 2 treatment allocation remains balanced, with each option having an equal probability of 0.25. As apparent in Tables \ref{positivity_assumption_1} and \ref{positivity_assumption_2}, violating the positivity assumption results in reduced prediction accuracy for the BAL approach, regardless of whether a DSS or a ISS prior is used. For example, we observe that the F$_1$ scores for both DSS and ISS are lower at stage 1, which is characterized by a highly imbalanced treatment allocation.\\
Although we did not explore this aspect in our study, we note that several strategies could be considered to mitigate the decline in predictive accuracy under treatment imbalance. For example, one could draw inspiration from oversampling techniques commonly used to address class imbalance in classification problems. In particular, methods such as the synthetic minority over-sampling technique \citep[SMOTE,][]{chawla2002} or random over-sampling examples \citep[ROSE,][]{menardi2014} could be adapted to artificially increase the representation of under-observed treatment groups. Concretely, one might generate synthetic pseudo-observations for individuals assigned to less frequent treatments, thus balancing the data prior to Bayesian inference. Alternatively, from a fully Bayesian standpoint, the Bayesian bootstrap \citep{rubin1981} provides a principled approach by assigning posterior weights to the observed data, effectively simulating additional information for underrepresented treatment groups.
\begin{table}[h]
\centering
\caption{Second experiment, imbalanced treatment assignment. For stage 2 and stage 1, three summaries of the accuracy in selecting significant variables (proportion of FN, proportion of FP, and F$_1$ score), for DSS and ISS, obtained by specifying $n=200$ and $T=4$. The number of individual covariates is $k=10$, while $\rho^*=0.6$. Results are based on 100 replicated datasets. }
\begin{tabular}{ccccccccc}
\Hline
\multicolumn{9}{c}{Stage 2}\\ \hline
 \multirow{2}{*}{$d_2$}   &\multicolumn{2}{c}{FN} & & \multicolumn{2}{c}{FP} & & \multicolumn{2}{c}{F$_1$ score}\\ \cline{2-3} \cline{5-6} \cline{8-9}
  & DSS & ISS  & & DSS & ISS & & DSS & ISS \\ \hline
  99 & 0.088 & 0.112 & & 0.088 & 0.066 & & 0.869 &0.877  \\ \hline
  \multicolumn{9}{c}{Stage 1}\\ \hline
 \multirow{2}{*}{$d_1$}   &\multicolumn{2}{c}{FN} & & \multicolumn{2}{c}{FP} & & \multicolumn{2}{c}{F$_1$ score} \\ \cline{2-3} \cline{5-6} \cline{8-9}
  & DSS & ISS  & & DSS & ISS & & DSS & ISS \\ \hline
  55 &0.175 &0.235 & &0.157 & 0.096 & & 0.765 & 0.775 \\ \hline
\end{tabular}
\label{positivity_assumption_1}
\end{table}
\begin{table}[h]
\centering
\caption{Second experiment, imbalanced treatment assignment. For stage 2, stage 1 and overall, two measures of prediction accuracy (MAE and ER), for DSS and ISS, obtained by specifying $n=200$ and $T=4$. The number of individual covariates is $k=10$, while $\rho^*=0.6$. Results are based on 100 replicated datasets. }
\begin{tabular}{cccccc}
\Hline
\multicolumn{6}{c}{Stage 2}\\ \hline
 \multirow{2}{*}{$d_2$}   &\multicolumn{2}{c}{MAE} & & \multicolumn{2}{c}{ER} \\ \cline{2-3} \cline{5-6} 
  & DSS & ISS  & & DSS & ISS  \\ \hline
  99 & 0.011 & 0.010 & &0.042 &0.042  \\ \hline
\multicolumn{6}{c}{Stage 1}\\ \hline
 \multirow{2}{*}{$d_1$}   &\multicolumn{2}{c}{MAE} & & \multicolumn{2}{c}{ER} \\ \cline{2-3} \cline{5-6} 
  & DSS & ISS  & & DSS & ISS  \\ \hline
 55 & 1.084  & 1.158 & & 0.294 & 0.307  \\ \hline
 \multicolumn{6}{c}{Overall}\\ \hline
 \multirow{2}{*}{$(d_1, d_2)$}   &\multicolumn{2}{c}{MAE} & & \multicolumn{2}{c}{ER} \\ \cline{2-3} \cline{5-6} 
  & DSS & ISS  & & DSS & ISS  \\ \hline
 (55, 99)  & 1.095 & 1.168 & & 0.324 & 0.336 \\ \hline
\end{tabular}
\label{positivity_assumption_2}
\end{table}

\subsection{Convergence of MCMC}
Assessing MCMC convergence is crucial for ensuring the reliability of posterior inference. To this end, we have included trace plots for key parameters---such as randomly selected non-zero regression coefficients and the hyperparameters $a$ and $b$, which govern the correlation in a regressor's significance across the two stages---as shown in Figure \ref{fig_traceplot}. Convergence diagnostics were also conducted after discarding the burn-in period (e.g., the first 5000 iterations), using the Geweke test \citep{geweke1992evaluating}, implemented via the \texttt{geweke.diag} function from the \texttt{coda} package in R. The trace plots demonstrate visually stable chains, and the corresponding Geweke's z-scores all lie within the range of -1.96 to 1.96, suggesting that the Markov chains have reached convergence.
 \begin{figure}[h]
  \centering
  \includegraphics[width=0.85\textwidth, height=7.4cm]{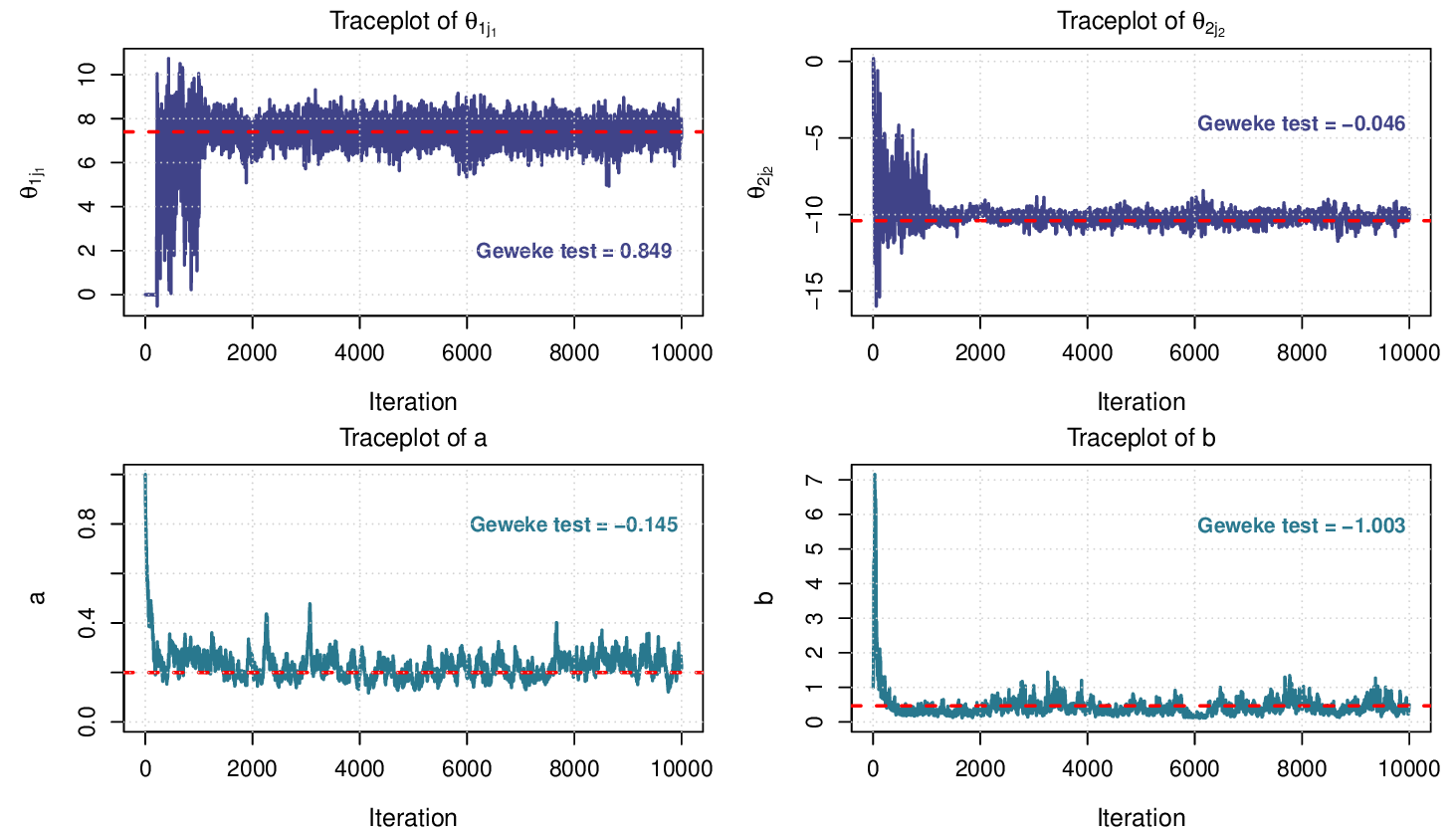}
  \caption{Top row: trace plots of two non-zero regression coefficients, $\theta_{1j_1}$ and $\theta_{2j_2}$, corresponding to two randomly selected indices $j_1$ and $j_2$. Bottom row: trace plots of the hyperparameters $a$ and $b$. All plots refer to a single replicated dataset from the first simulation study, under the scenario with $n = 25$, $k = 10$, and $\rho^* = 0.6$. The red lines denote the true values. Each panel displays the corresponding Geweke's z-score, with values falling within the range of $(-1.96,1.96)$ generally considered evidence of convergence.}
  \label{fig_traceplot}
\end{figure}

\subsection{Computation time}
Table \ref{table_computation_time} reports the computation times (in minutes) for the proposed Bayesian methods (DSS and ISS), as well as Q-learning with lasso, Q-learning with RF, and AOWL. These results are based on varying sample sizes and covariate dimensions to illustrate the scalability and computational cost of each approach. The Bayesian methods require longer computation times, particularly as the sample size increases.
\begin{table}
\caption{Computation times (in minutes) for DSS, ISS, Q-learning with lasso (QL), Q-learning with RF (QRF), and AOWL, referring to the total time required to analyze 100 datasets from the first simulation experiment, with $(k, n) \in \{(10,25), (20,50), (30,75)\}$
and $\rho^* = 0.6$. Times for DSS and ISS correspond to running Algorithm \ref{algo} for 10,000 iterations on a machine equipped with an Apple M3 chip and 8GB of unified memory.}\label{lpml clinical}
\begin{center}
\begin{tabular}{lcccccc}
\Hline
n   & (d1,d2) & DSS      & ISS      & QL   & QRF  & AOWL  \\ \hline
25  & (22, 33)  & 17.246 & 17.113 & 0.400 & 0.086 & 0.118 \\
50  & (42, 63)   & 28.743  & 28.706 & 0.615 & 0.281 & 0.125 \\
75  & (62, 93)   & 41.675 & 41.611  & 0.853 & 0.601 & 0.155 \\ \hline
\end{tabular}
\label{table_computation_time}
\end{center}
\end{table}

\section{Additional Results on the Clinical Case Study}\label{app:clinical}
We present additional results from the MIMIC-III data analysis, obtained by fitting DSS, ISS, Q-learning with lasso and Q-learning with RF.\\
Table \ref{fren analysis} reports the frequencies of  treatments actually administered in two stages in the
dataset. Table \ref{lpml clinical} compares the goodness-of-fit of the DSS and ISS approaches using the Bayesian Information Criterion (BIC) \citep{SCH78} and the Log Pseudo Marginal Likelihood (LPML) \citep{geisser1979predictive}, indicating a superior fit for the DSS approach. 
Table \ref{frequency all} reports the frequency of optimal treatment selections by DSS, ISS, Q-learning with lasso, and Q-learning with RF, across the two stages.
Table \ref{treat all} reports the optimal treatment selections for two randomly chosen patients, as determined by the four models across both stages.\\

Finally, since the posterior samples from Algorithm \ref{algo} allow us to compute the marginal probabilities of each antihypertensive agent being included in the optimal treatment, Figure \ref{second patient combination_new_2} presents these probabilities for the same two patients shown in Figure \ref{second patient combination_new_1}.
\begin{table}[H]
\centering
\caption{MIMIC-III data. Frequencies of actual treatments administered in two stages in the dataset. A refers to ACEi, B refers to beta-blockers, C refers to CCB, and D refers to diuretics.}\label{fren analysis}
\begin{tabular}{cccccccccc}
\Hline
        & A    & B    & C   & D    & AB  & AD  & BD  & ABD & Total \\
        \hline
Stage 2 & 23   & 60   & 15  & 143  & 28  & 14  & 31  & 6   & 320   \\
(\%)    & 7.2  & 18.8 & 4.7 & 44.7 & 8.8 & 4.4 & 9.7 & 1.9 &       \\
Stage 1 & 108  & 165  & 41  & 239  & 66  & 28  & 71  & 24  & 742   \\
(\%)    & 14.6 & 22.2 & 5.5 & 32.2 & 8.9 & 3.8 & 9.6 & 3.2 &   \\ \hline   
\end{tabular}
\end{table}

\begin{table}[H]
\centering
\caption{MIMIC-III data. Values taken by BIC and LPML for DSS, ISS.}\label{lpml clinical}
\begin{tabular}{cccccc}
\Hline
    & \multicolumn{2}{c}{Stage 2}      &       & \multicolumn{2}{c}{Stage 1}        \\ \hline 
    & DSS               & ISS       && DSS               & ISS    \\ \cline{2-3} \cline{5-6}
BIC &4681.889  &5023.239 
& &7023.151  &7160.299 \\ 
\multicolumn{1}{c}{LPML} &-2493.087       &-2510.148   
& &-3490.061      &-3495.198  
\\ \hline 
\end{tabular}
\label{bic and lpml}
\end{table}

\begin{table}[H]
\centering
\caption{Frequency of optimal treatment selections by DSS, ISS, Q-learning with lasso (QL), and Q-learning with RF (QRF) across two stages (stage 1: $n=742$; stage 2: $n=320$).}  
\begin{tabular}{cccccccccc}
\Hline
\multicolumn{1}{l}{Methods}            &                    & A               & B               & C               & D               & AB             & AD   & BD   & ABD \\  \hline 
\multirow{4}{*}{DSS}                   & Stage 2            & 86              & 125             & 86              & 66              & 25             & 198  & 45   & 111  \\
                                       & \%                 & 0.12            & 0.17            & 0.12            & 0.09            & 0.03           & 0.27 & 0.06 & 0.15 \\
                                       & Stage 1            & 59              & 39              & 32              & 114             & 38             & 2    & 26   & 10   \\
                                       & \%                 & 0.18            & 0.12            & 0.10            & 0.36            & 0.12           & 0.01 & 0.08 & 0.03 \\ \hline 
\multirow{4}{*}{ISS}                   & Stage 2            & 143             & 44              & 71              & 169             & 94             & 77   & 65   & 79   \\
                                       & \%                 & 0.19            & 0.06            & 0.10            & 0.23            & 0.13           & 0.10 & 0.09 & 0.11 \\
                                       & Stage 1            & 34              & 5               & 36              & 95              & 41             & 26   & 39   & 44   \\
                                       & \%                 & 0.11            & 0.02            & 0.11            & 0.30            & 0.13           & 0.08 & 0.12 & 0.14 \\ \hline 
\multirow{4}{*}{QL}                   & Stage 2            & 251             & 33              & 59              & 118             & 10             & 220  & 45   & 6    \\
                                       & \%                 & 0.34            & 0.04            & 0.08            & 0.16            & 0.01           & 0.30 & 0.06 & 0.01 \\
                                       & Stage 1            & 99              & 104             & 19              & 25              & 9              & 22   & 23   & 19   \\
                                       & \%                 & 0.31            & 0.32            & 0.06            & 0.08            & 0.03           & 0.07 & 0.07 & 0.06 \\ \hline 
\multirow{4}{*}{QRF}                  & Stage 2            & 220             & 154             & 92              & 122             & 64             & 58   & 22   & 10   \\
                                       & \%                 & 0.30            & 0.21            & 0.12            & 0.16            & 0.09           & 0.08 & 0.03 & 0.01 \\
                                       & Stage 1            & 130             & 66              & 16              & 68              & 32             & 2    & 5    & 1    \\
                                       & \%                 & 0.41            & 0.21            & 0.05            & 0.21            & 0.10           & 0.01 & 0.02 & 0.00 \\ \hline 
\end{tabular}
\label{frequency all}
\end{table}

\begin{table}[H]
\centering
\caption{Optimal treatments selected in two stages by DSS, ISS, Q-learning with lasso (QL), and Q-learning with RF (QRF) for the same two patients identified in Figure \ref{second patient combination_new_1}.} 
\begin{tabular}{llcc}
\Hline
                & Methods                  & Stage 2                                       & Stage 1                                 \\ \hline
                     \multirow{4}{*}{ \# individual 1 }             & DSS                      & n.a.                                           & ACEi                                    \\
                                 & ISS                      & n.a.                                           & ACEi+diuretics                          \\
                                 & QL                      & n.a.                                          & ACEi                                    \\
                                 & QRF                     & n.a.                                          & beta-blockers                           \\
 \hline
               \multirow{4}{*}{ \# individual 2 }                     & DSS                      & diuretics                                     & beta-blockers+diuretics                 \\
                                 & ISS                      & ACEi+beta-blockers+diuretics                  & CCB                                     \\
                                 & QL                      & ACEi                                          & beta-blockers                           \\
                                 & QRF                     & diuretics                                     & ACEi    \\                               \hline
\end{tabular}
\label{treat all}
\end{table}

\begin{figure}[h]
\centering
\includegraphics[clip,trim=0.6cm 1.5cm 1.0cm 1.5cm,width=0.3\textwidth]{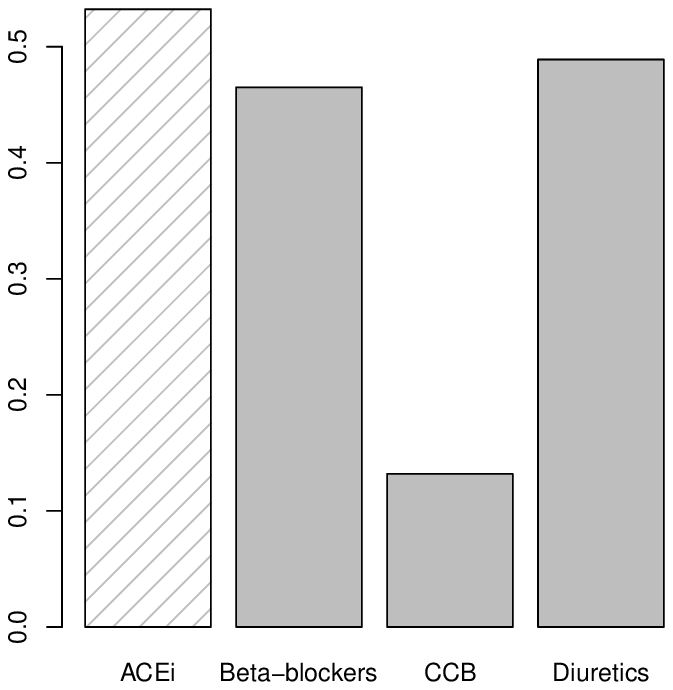}\hspace{0.3cm}
\includegraphics[clip,trim=0.6cm 1.5cm 1.0cm 1.5cm,width=0.3\textwidth]{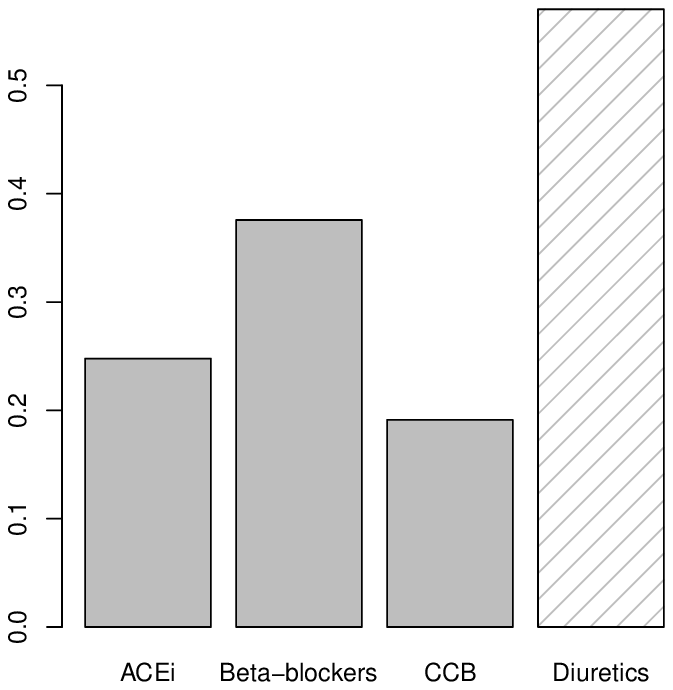}\hspace{0.3cm}
\includegraphics[clip,trim=0.6cm 1.5cm 1.0cm 1.5cm,width=0.3\textwidth]{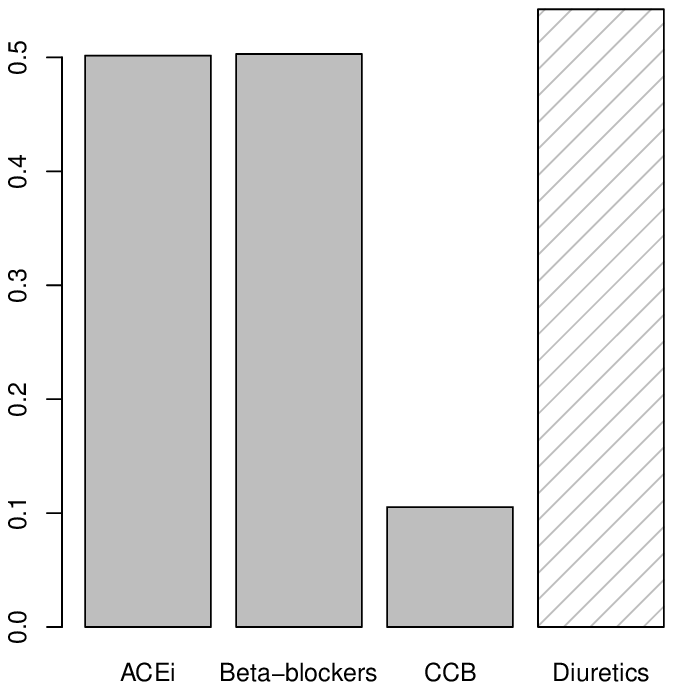} 
\caption{MIMIC-III data. Posterior marginal probabilities, estimated by DSS, for each of the four antihypertensive agents being included in the optimal treatment combination: Individual \#1 at stage 1 (left panel), Individual \#2 at stage 2 (middle panel), and Individual \#2 at stage 1 (right panel). In each bar plot, the white segment highlights the agent with the highest posterior marginal probability.}
\label{second patient combination_new_2}
\end{figure}

\end{appendices}

\clearpage
\bibliographystyle{apalike}
\bibliography{references}

\end{document}